\documentclass{article}

\usepackage[utf8]{inputenc}
\usepackage[T1]{fontenc}

\usepackage{amsmath, amssymb}
\usepackage{amsthm}
\usepackage[disable]{todonotes}
\usepackage{thmtools}
\usepackage{thm-restate}
\newtheorem{theorem}{Theorem}
\newtheorem{lemma}{Lemma}

\theoremstyle{remark}
\newtheorem*{remark}{Remark}

\usepackage{mathrsfs}
\usepackage{float}
\usepackage{mathtools}
\usepackage{graphicx}
\usepackage{parskip}
\usepackage{hyperref}
\usepackage{xcolor}
\usepackage{subcaption}
\usepackage{booktabs}
\usepackage{caption}
\newtheorem{axiom}{Axiom}[section]

\title{Minimal Effort to Consensus (MEC) polarization measure}
\author{
  Jes\'{u}s Aranda\\Universidad del Valle
  \and
  Juan Francisco D\'{\i}az\\Universidad del Valle
  \and
  Juan Camilo Narv\'{a}ez\\Universidad del Valle
  \and
  Catuscia Palamidessi \\ INRIA-Saclay
  \and
  Carlos Pinz\'{o}n\\INRIA-Saclay
  \and
  Frank Valencia\\CNRS-LIX, \'Ecole Polytechnique
  \and
  Oscar Vargas\\Universidad Javeriana Cali
}

\begin{document}
\newtheorem{definition}{Definition}
\newtheorem{example}{Example}

\newtheorem{proposition}{Proposition}

\newcommand{\Reals}{\mathbb{R}}
\newcommand{\Nats}{\mathbb{N}}
\newcommand{\bvec}[1]{\boldsymbol{#1}}
\newcommand{\Dset}{\mathcal{D}}
\newcommand{\Pol}{\mathbf{P}}

\newcommand{\Argmax}{\text{argmax}}
\newcommand{\PositiveReals}{\Reals_{\geq 0}}

\newcommand{\Mass}[1]{\Sigma #1}
 \newcommand{\defsymbol}{\stackrel{\text{def}}{=}}

\newcommand{\support}{\boldsymbol{s}}

\def\oo{\infty}
\def\eps{\epsilon}
\def\calX{{\mathcal{X}}}


\NewDocumentCommand{\ER}{O{\alpha}}{\operatorname{ER}_{#1}}

\NewDocumentCommand{\ERLabel}{o}{\ensuremath{\mathrm{ER}\IfValueT{#1}{_{#1}}}}
\newcommand{\EMDLabel}{\ensuremath{\mathrm{EMD}}}
\newcommand{\VdELabel}{\ensuremath{\mathrm{VdE}}}
\newcommand{\TWDLabel}{\ensuremath{\mathrm{TWD}}}
\newcommand{\KILabel}{\ensuremath{\mathrm{KI}}}

\NewDocumentCommand{\EC}{O{\alpha,\beta} O{\pi}}{%
\operatorname{EC}_{#1}^{#2}%
}

\NewDocumentCommand{\MEC}{O{\alpha,\beta}}{\operatorname{MEC}_{#1}}

\NewDocumentCommand{\ECEMD}{O{\alpha,\beta} O{\pi}}{%
\operatorname{EC}^{#2}%
}

\NewDocumentCommand{\MECEMD}{O{\alpha,\beta}}{\operatorname{MEC}}

\def\EMD{\operatorname{EMD}}

\NewDocumentCommand{\set}{s m}{\IfBooleanTF{#1}{{\{#2\}}}{{\ensuremath{\left\{#2\right\}}}}}
\NewDocumentCommand{\lr}{s m}{\IfBooleanTF{#1}{{(#2)}}{{\ensuremath{\left(#2\right)}}}}
\NewDocumentCommand{\eqdef}{}{\stackrel{{\scriptscriptstyle \mathrm{def}}}{=}}
\NewDocumentCommand{\argmin}{}{\operatorname*{arg\,min}}
\NewDocumentCommand{\bpi}{}{{\boldsymbol{\pi}}}
\NewDocumentCommand{\bx}{}{\boldsymbol{x}}
\NewDocumentCommand{\midrange}{}{{\sigma_\oo}}
\NewDocumentCommand{\lrfrac}{m m}{\lr{\frac{#1}{#2}}}

\newcommand{\doubleCheckmark}[2][1]{%
  \scalebox{#1}{\checkmark\hspace{-#2}\textcolor{white}{\checkmark}\hspace{-#2}\checkmark}%
}
\def\checkcheck{{\doubleCheckmark{1em}}}

\newcommand{\Var}{\mathrm{Var}}

\newcommand{\Disc}{\Delta
}

\newcommand{\proofpart}[1]{\par\medskip\noindent\textbf{#1.}\ }

\maketitle

\section{Introduction}

Polarization is a central feature of contemporary public discourse.
It appears in electoral competition, climate-policy debates, and
attitudes toward institutions, where populations are often described
as divided into mutually antagonistic camps. To study such phenomena,
one first needs a measure. Claims about the causes of polarization,
its persistence, and its possible reduction depend on what the measure
is taken to express. The problem is therefore not merely to assign a
number to a distribution of opinions, but to say what that number
means.

Several answers have been proposed. The Esteban--Ray measure
(\ERLabel)~\cite{esteban1994measurement} defines polarization
through a pairwise sum combining intra-group identification and
inter-group alienation. The Earth Mover's Distance
(\EMDLabel), from optimal-transport
theory~\cite{geher2020isometric,vallender1974calculation}, measures
the minimum cost of transforming one distribution into another.
The Tastle--Wierman Dissention measure
(\TWDLabel)~\cite{tastle2007consensus} adapts
information-theoretic ideas to ordinal scales. The Van der Eijk index
(\VdELabel)~\cite{van2001measuring} treats polarization as the
complement of agreement under a unimodality criterion. The
Koudenburg--Kiers--Kashima Index
(\KILabel)~\cite{koudenburg2021new} is an empirical
regression-based measure fitted to expert judgments of polarization.
These measures
differ in motivation and form, but they share a common feature:
polarization is computed directly from the distribution, either as a
pairwise aggregate, an ordinal dispersion index, or a transport cost.
They do not, however, identify a point at which consensus could be
reached with minimum effort, nor do they provide the corresponding
optimization structure.

We propose such a structure. Our starting point is the following
interpretation: polarization is resistance to consensus. A population
is highly polarized when much effort is needed to bring its members to
a common position; it is weakly polarized when little effort is needed.
This idea leads to an optimization problem. Given an opinion
distribution, we define its polarization as the minimum effort required
to transform it into a consensus distribution. We call this quantity
the \emph{Minimum Effort to Consensus} (MEC). The minimum is taken over
all possible consensus points. A minimizer, denoted \(y^*\), is
therefore not an external parameter, but an output of the measure: it
is the consensus point that is least costly to reach.

The definition is connected to transportation theory. In the basic
case, where effort is mass times distance, MEC is the
\(1\)-Wasserstein distance from \(\pi\) to its nearest consensus
distribution, equivalently the Earth Mover's Distance
between \(\pi\) and the least costly consensus configuration. More
generally, for parameters \(\alpha,\beta\geq 1\), MEC is the minimum,
over consensus points, of a weighted \(L^\beta\) cost in which the
weights are the \(\alpha\)-power of the population masses. This
parametric form includes two familiar dispersion measures as special
cases: mean absolute deviation about a median when
\((\alpha,\beta)=(1,1)\), and variance-like dispersion when
\((\alpha,\beta)=(1,2)\). The parameters also have a natural social
interpretation. The exponent \(\alpha\) controls the effect of group
size and corresponds to the identification dimension emphasized by
Social Identity Theory~\cite{TajfelTurner1979}; the exponent \(\beta\)
controls the cost of large opinion shifts and corresponds to the
alienation dimension, related to the cognitive cost of substantial
belief revision~\cite{Festinger1957}.

The contributions of the paper are as follows.

\paragraph{Conceptual foundation.}
We define polarization as minimum effort to consensus. Formally, this
amounts to measuring the distance from an opinion distribution to the
set of consensus distributions. The definition yields both a scalar
polarization value and an optimal consensus point. It also recovers
classical dispersion measures in special cases, while allowing two
parameters corresponding to identification and alienation.

\paragraph{Theoretical results.}
We establish a geometric characterization of MEC in the transport
case. In particular, Theorem~\ref{th:mec-as-emd} shows that MEC can be
expressed as \(\Mass{\pi}/2\) minus the Earth Mover's Distance from
\(\pi\) to the extremal distribution. Thus, in this case,
polarization is proximity to maximum disagreement. The optimization
formulation also yields structural information beyond the scalar
value. The optimal point \(y^*\) identifies the consensus position that
minimizes total effort. We prove a Minority Principle
(Theorem~\ref{thm:minority-principle-general}), according to which, in
a two-group distribution, the smaller group bears the greater share of
the effort to consensus. We also introduce a Tipping Point method
(Section~\ref{sec:tipping-point}) that gives, in closed form, the
critical mass at which moving opinions toward an extreme begins to
increase rather than decrease polarization. Together, these results
reveal that polarization is not monotone in extremism, nor even in
movement away from the optimal consensus point. Displacing the whole
of a group away from the optimal consensus point raises polarization,
but transferring only a portion of the group can lower it: the detached
fragment may be too small to exert influence at its new position,
while the original group is weakened. The Tipping Point method
makes this dependence explicit and computable. Finally,
we show that MEC
satisfies the three axioms of Esteban and Ray~\cite{esteban1994measurement},
as well as a central-split monotonicity property, under natural
conditions on the parameters.

\paragraph{Empirical and computational results.}
We compare MEC with the expert-judgment benchmark of Koudenburg,
Kiers, and Kashima~\cite{koudenburg2021new}, in which sixty researchers in
polarization studies rated fifteen hypothetical opinion distributions.
The parametrization \(\MEC[2,1.15]\) obtains Kendall's
\(\tau\approx 0.89\) against the expert ranking. This matches the
strongest parametrization of \ERLabel{} and outperforms
the dispersion-based alternatives \VdELabel{} and
\TWDLabel. MEC is also computationally
tractable. For support size \(n\), the convex objective can be
minimized by bisection in \(O(n\log(1/\varepsilon))\) time to precision
\(\varepsilon\). This is asymptotically faster than the \(O(n^2)\)
cost of pairwise measures. The difference is small on short ordinal
scales, but becomes relevant for fine-grained opinion spaces and
large-scale corpus analyses such as those in
Section~\ref{sec:experiments}.

MEC and the Esteban--Ray measure are conceptually different.
\ERLabel{} measures pairwise antagonism between groups. MEC measures
the resistance of the whole distribution to being brought into
agreement. The two measures coincide in spirit on canonical cases,
such as the symmetric two-extreme distribution, but MEC adds an
optimization structure and an endogenous consensus point while
remaining compatible with the axiomatic perspective introduced by
Esteban and Ray.

The remainder of the paper is organized as follows.
Section~\ref{sec:preliminaries} fixes notation and defines opinion
distributions and polarization measures. Section~\ref{sec:mec}
introduces MEC, motivates the parameters \(\alpha\) and \(\beta\), and
discusses computation. Section~\ref{sec:structural} proves structural
properties, including convexity, the \(L^\beta\)-dispersion
characterization, symmetry, and population scaling.
Section~\ref{sec:theoretical-applications} presents the Minority
Principle and the Tipping Point method.
Section~\ref{sec:shifting-away} proves that MEC is a polarization
measure in the standard axiomatic sense.
Section~\ref{sec:er-comparison} compares MEC with the Esteban--Ray
measure and proves axiomatic compatibility.
Section~\ref{sec:experiments} reports the empirical validation.
Section~\ref{sec:conclusion} concludes and situates the contribution
in related work.
Proofs are collected in the appendices.

\section{Preliminaries}\label{sec:preliminaries}

We will define a family of polarization measures for populations of individuals (agents) with opinions on some underlying proposition (e.g., \emph{climate change is a hoax}). These opinions are represented as numbers in $[0,1]$, where a value of $0$ (a value of $1$) indicates complete disagreement (complete agreement) with the proposition, and the higher the value, the stronger the agreement.

We shall use mass \emph{distributions} with finite support to represent how agents are grouped according to their opinions. More precisely,

\begin{definition}[Opinion Distribution]  
An \emph{opinion distribution} is a function $\pi : [0,1] \to \PositiveReals$ with finite, non-empty support $\support(\pi)=\{x \in [0,1] \mid \pi(x) > 0\}$. For each $x \in \support(\pi)$, the value $\pi(x)$ is the \emph{weight} (or \emph{mass}) of the group of agents holding opinion $x$. The (total) \emph{mass} (or \emph{total weight}) of $\pi$ is $\Mass{\pi} \eqdef \sum_{x \in \support(\pi)} \pi(x)$. We write $\Dset$ for the set of all opinion distributions, and $\Dset_k$ for those with mass $k$.  
\end{definition}

We shall often use $(\bvec{\pi},\bvec{x})=((\pi_1,\ldots,\pi_n),(x_1,\ldots,x_n)) \in \Reals_{\geq 0}^n \times [0,1]^n$, with $x_1 < \cdots < x_n$ and $\sum_{i=1}^n \pi_i > 0$, to denote the unique opinion distribution $\pi$ satisfying $\pi(x_i)=\pi_i$ for each $i$ and $\Mass{\pi}=\sum_{i=1}^n \pi_i$.  

Notice that opinion distributions do not need to be normalized (a probabilistic distribution). Thus, we define the following. 

\begin{definition}[Power and Normalized Distributions]
Let $\pi$ be an opinion distribution. The $\alpha$-power opinion distribution of $\pi$ is given by  $\pi^\alpha(x) \defsymbol \pi(x)^\alpha$ for every $x\in[0,1]$ where $\alpha \in \PositiveReals.$ 
We shall use $\overline{\pi}(x)\defsymbol \frac{\pi(x)}{\Mass{\pi}}$ for every $x \in [0,1]$ for the normalized opinion distribution obtained from $\pi$. 
\end{definition}

It will also be convenient to recall the notion of median for an opinion distribution. A real number $m\in[0,1]$ is a \emph{median} of $\pi\in\Dset$ if $\sum_{x\le m}\pi(x)\ge \Mass{\pi}/2$ and $\sum_{x\ge m}\pi(x)\ge \Mass{\pi}/2$. Medians need not be unique; the set of medians of $\pi$ is a closed subinterval of $[0,1]$.

Let us now introduce some representative families of one-point and two-point distributions for polarization. An extremal distribution is a two-point distribution placing half of its total mass at the minimum opinion value (0) and the other half at the maximum opinion value (1). More precisely, 

\begin{definition}[Consensus and Extremal Distributions]  
An opinion distribution $\pi$ is \emph{(two-point) extremal} if $\pi(0)=\pi(1)=\Mass{\pi}/2$. For any $\pi$, we write $\pi^*$ for the extremal distribution with the same mass, i.e.\ $\Mass{\pi^*}=\Mass{\pi}$. An opinion distribution $\pi$ is a \emph{consensus (distribution) at $y \in [0,1]$} if $\support(\pi)=\{y\}$, and a \emph{consensus} if this holds for some $y \in [0,1]$. If $\pi \in \Dset$, we write $\pi_y$ for the consensus distribution at $y$ with mass $\Mass{\pi_y}= \Mass{\pi}$.
\end{definition}

We can now introduce the notion of polarization measure.

\begin{definition}[Polarization Measure]\label{pol:def}  
A \emph{polarization measure} is a function $\Pol:\Dset \to \Reals_{\geq 0}$ such that:  
(1) if $\pi \in \Dset$ is a consensus then $\Pol(\pi)=0$; and 
(2) if $\pi \in \Dset_k$ is extremal, then $\Pol(\pi) \geq \Pol(\pi')$ for all $\pi' \in \Dset_k$.  
\end{definition}

The two basic conditions in the definition above are intuitive and standard properties in the literature \cite{esteban1994measurement} for polarization measures: if there is consensus, there is no polarization, and extremal distributions maximize polarization among distributions with the same total mass.

\section{The MEC Measure}\label{sec:mec}

In this section we introduce our notion of polarization and establish
some of its properties. For clarity of exposition, we begin with a
simpler formulation, which will later be generalized to obtain the
measure presented in this paper.

As argued in the Introduction, the greater the polarization of a
population, the more difficult it is to achieve consensus. It is
therefore natural to define polarization in terms of the \emph{minimal
effort} required to reach consensus. Inspired by transportation
theory, we can measure the effort required to transform an opinion
distribution $\pi$ into a consensus at $y$ by
\begin{equation}\label{eq:emd}   
\ECEMD(y)\defsymbol\sum_{x \in \support(\pi)}\pi(x) |x - y|.
\end{equation}
Thus, the minimal effort to a consensus (MEC) of $\pi$ would be:  
\begin{equation}\label{eq:mecemd}     
\MECEMD(\pi) \defsymbol  \min_{y\in[0,1]} {\ECEMD}(y)=\min_{y\in[0,1]}\sum_{x \in \support(\pi)}\pi(x) |x - y|.
\end{equation}

The quantity $\ECEMD(y)$ is precisely the \emph{Earth Mover's Distance} (\EMDLabel) between
$\pi$ and the consensus distribution $\pi_y$; that is,
\[
\ECEMD(y) = \EMD(\pi,\pi_y).
\]
\todo{FV: Juan Camilo please double check this} A formal definition of \EMDLabel{} on opinion distributions, together with a proof of this identity, is given in Appendix~\ref{appendix:emd-definition} (Lemma~\ref{lem:ecemd-as-emd}). Intuitively, \EMDLabel{} represents the least amount of work required
to transform one distribution into another. By Cor.~\ref{cor:mec-special-cases} in Sec.~\ref{section:special-cases}, this work is minimal when 
$y$ is a median $m$ of $\pi$. Hence,
\begin{equation}
\MECEMD(\pi) = \EMD(\pi,\pi_m).\label{eq:mec-emd-pi-m}
\end{equation}

For example, consider the distribution 
$\pi=((20,30,50),(0,0.8,1))$. 
Reaching consensus at the median $y=0.8$ requires an effort of $
\ECEMD(0.8)=26,
$  
and this value is minimal among the efforts required to reach a consensus. 
Hence,
$
\MECEMD(\pi)=\ECEMD(0.8)=26
= \EMD(\pi,\pi_{0.8}).
$

Now consider the extremal distribution 
$\pi^*=((50,50),(0,1))$. 
In this case,
$
\MECEMD(\pi^*)=50
$. In general, for any $\pi$, \[ \MECEMD(\pi^*)=\frac{\Mass{\pi}}{2}\] 
can be shown to be the maximal possible value of $\MECEMD$ among 
distributions of total mass $\Mass{\pi}$, i.e., $\pi^*$ is the most polarized distribution of this mass. Furthermore, it is easy to see that if $\pi$ is a consensus distribution then $\MECEMD(\pi)=0$. Thus, according to Def.~\ref{pol:def}, $\MECEMD(.)$ is indeed a polarization measure.

From the above observations and properties of \EMDLabel, we can derive the following elegant characterization for polarization in terms of the distance to the extremal distribution.

\begin{restatable}[MEC as \EMDLabel\ distance to $\pi^*$]{theorem}{mecAsEmd}\label{th:mec-as-emd}
Let $\pi$ be an opinion distribution. Then  \begin{equation}
\MECEMD(\pi)=\frac{\Mass{\pi}}{2} - \EMD(\pi,\pi^*).    
\label{eq:mec-emd-1}
\end{equation}
\end{restatable}

The above proposition makes precise the intuition that proximity to the 
extremal distribution corresponds to greater polarization.

\subsection{Group Identity ($\alpha$) and Alienation ($\beta$)}
As argued in \cite{esteban1994measurement}, polarization increases with the degree of homogeneity within groups and with the degree of heterogeneity across groups. 
Greater within-group homogeneity reflects stronger \emph{group identity} or cohesion, while greater between-group heterogeneity captures increased \emph{alienation} across groups. 
If these two dimensions are to be adequately represented, the measure $\MECEMD(\cdot)$ in Eq.~\ref{eq:emd} requires reconsideration.

\emph{Alienation.}
Observe first that opinion distance in Eq.~\ref{eq:emd} enters linearly through the term $|x-y|$, a feature it shares with the classical formulation of Esteban and Ray \cite{esteban1994measurement}. 
Under this specification, a movement from $0$ to $1$ is exactly twice as costly as a movement from $0$ to $0.5$. 
Opinion change is thus treated as proportional to ideological distance, much as the transportation of mass across space.

There is, however, reason to doubt that ideological displacement proceeds in so uniform a manner. 
According to Social Identity Theory, individuals derive part of their self-conception from membership in social groups, and opinions frequently serve as outward expressions of such affiliation \cite{TajfelTurner1979}.
A modest adjustment within a shared ideological region may leave group attachment undisturbed. 
A more substantial shift, by contrast, may be interpreted as a departure from one's group or an alignment with a rival one, thereby engaging considerations of identity. 
In addition, significant revisions of belief may give rise to cognitive dissonance which individuals are inclined to avoid \cite{Festinger1957}.
It is therefore plausible that large ideological displacements entail costs that increase more than proportionally with distance.

To allow for this possibility, we introduce a parameter $\beta \geq 1$ and replace the term $|x-y|$ in Eq.~\ref{eq:emd} by $|x-y|^{\beta}$. 
When $\beta=1$, the proportional structure is preserved. 
When $\beta>1$, larger displacements receive disproportionately greater weight. 
The parameter $\beta$ may thus be understood as regulating alienation: it determines the rate at which the cost of opinion change increases with ideological distance.

\begin{example} Consider the two-point distribution
\(
\pi = ((49,52),(0,1)).
\)
Assume $\beta = 1$.  Since the median of $\pi$ is $m=1$, the minimal effort to consensus is obtained by moving the entire minority mass from $0$ to $1$, yielding a cost of $49$.
Though mathematically correct, this conclusion may appear sociologically doubtful: a minority comprising $49\%$ of the population would not readily be expected to adopt the opposite extreme, but rather to sustain some intermediate position.

If instead we set $\beta=2$, the objective becomes strictly convex and is minimized at the weighted mean (see Cor.~\ref{cor:mec-special-cases}), here $y^*=\frac{2704}{5105}\approx 0.5297.$
The optimal consensus thus lies closer to the midpoint rather than at an extreme.
\qed
\end{example}

\emph{Group Identity.} To address group identity in polarization, we also introduce a parameter $\alpha \geq 1$  and replace the mass $\pi(x)$ by $\pi(x)^\alpha$ in Eq.~\ref{eq:emd} to modulate the influence of group size. When $\alpha = 1$, a group's contribution is proportional to its mass.
When $\alpha > 1$, larger groups receive more-than-proportional weight, while smaller groups are relatively downweighted.
Thus increasing $\alpha$ amplifies the influence of cohesive majorities and attenuates the impact of negligible groups, mirroring the identification principle emphasized by \cite{esteban1994measurement}.

To illustrate the role of group identity, we write \(\MEC(\cdot)\) for the variant of \(\MECEMD(\cdot)\) in Eq.~\ref{eq:mecemd} obtained by replacing \(\pi(x)\) with \(\pi(x)^\alpha\) and \(|x-y|\) with \(|x-y|^\beta\). A precise definition of \(\MEC(\cdot)\) will be given in the next section.

\begin{example}\label{ex:group-identity}
Assume $\beta=2$ and consider
\[
\pi=((40,60),(0.25,0.75)),
\qquad
\pi'=((40,59,1),(0.25,0.75,1)).
\]
The distribution $\pi$ represents two moderate groups of unequal size. 
The distribution $\pi'$ may be viewed as arising from $\pi$ through a slight radicalization: a small minority moves toward the extreme right, while the majority remains moderate.

If $\alpha=1$, as shown below in Cor.~\ref{cor:mec-special-cases}, the optimal consensus is attained at the mean of $\pi$
($y=0.55$ for $\pi$ and $y=0.5525$ for $\pi'$). 
Because of the small extreme mass, however, we obtain
\[
\MEC[1,2](\pi)=6.0 \;<\; 6.161875 \;=\; \MEC[1,2](\pi'),
\]
so that $\pi'$ appears more polarized.

If instead $\alpha=2$, the ordering reverses. 
In this case, from Cor.~\ref{cor:mec-special-cases}, the optimal consensus points are at the mean of $\overline{\pi^\alpha}$ and $\overline{\pi'^\alpha}$, respectively:
\[
y=\frac{31}{52}\approx 0.59615 \quad \text{for }\pi,
\qquad
y'=\frac{1721}{2904}\approx 0.59263 \quad \text{for }\pi',
\]
and we obtain
\[
\MEC[2,2](\pi)=\frac{3600}{13}\approx 276.923 \;>\; 274.207 \approx \frac{3185183}{11616} \;=\; \MEC[2,2](\pi').
\]
In this case, the larger cohesive groups dominate the polarization structure, and the small extremist minority exerts comparatively little influence.
\qed
\end{example}

The above example shows that, for suitable values of the identification
and alienation parameters (here $\alpha=\beta=2$), transferring a small
portion of a group toward the extremes does not necessarily increase
polarization. This accords with the thesis of Esteban and Ray
\cite{esteban1994measurement} that small groups contribute little to
overall polarization. Indeed, such a shift may diminish the effective
weight of the original bloc, and thereby reduce the overall level of
polarization.

\subsection{Generalized MEC}
The discussion in the previous section leads us to the following parametric definition of MEC.

\begin{definition}[Generalized MEC]\label{mec:def} Let $\alpha,\beta \in [1,\infty)$, $y \in [0,1]$ and let $\pi \in \Dset$ be an opinion distribution. 
The $(\alpha,\beta)$-\emph{effort to a consensus at $y$} for $\pi$ is determined by the function   $\EC:[0,1]\to \PositiveReals$ given by \begin{equation}    
\EC(y)\defsymbol\sum_{x \in \support(\pi)}\pi(x)^\alpha |x - y|^\beta.
\end{equation}
The \emph{minimal  $(\alpha,\beta)$-effort to consensus} for $\pi$ is determined by the function $\MEC: \Dset\to\Reals_{\geq 0}$ given by \begin{equation}    
\MEC(\pi) \defsymbol  \min_{y\in[0,1]} {\EC}(y).
\end{equation}
 \end{definition}
Observe that \(\MEC(\pi)\) is well defined: Since the effort function \(\EC\) is continuous on the compact interval \([0,1]\), it must attain a minimum by the extreme value theorem. The points at which \(\EC\) attains its minimum will be called the $(\alpha,\beta)$-\emph{optimal consensus points} of $\pi$.  It follows from the strict convexity result in Th.~\ref{convexity} in the next section that, when \(\beta>1\), there is but one $(\alpha,\beta)$-optimal consensus point. We shall omit the prefix $(\alpha,\beta)$- when it is understood from the context. 

\begin{remark}
In what follows, we shall often use the term MEC to denote the general quantity $\MEC$ introduced in Definition~\ref{mec:def}, rather than the particular instance $\MEC[1,1]$ given in Eq.~\ref{eq:mecemd}, which was presented for motivational purposes.
\end{remark}

\section{Structural Properties of MEC}\label{sec:structural}

In this section we state some foundational properties of MEC such as its convexity,  $L^p$-dispersion characterization, symmetry and scaling.

\subsection{Convexity}
We now state a useful feature of our effort functions: They are \emph{convex}.  From a theoretical perspective, convexity provides a solid theoretical framework. As shown here, it allows us to use well-established results from convex analysis (e.g., \cite{Rockafellar1970}) to further analyze the behavior of and to derive additional properties of MEC. 

\begin{restatable}[Convexity of Effort]{theorem}{convexityThm}\label{convexity}
The effort function 
$\EC:[0,1] \to \PositiveReals$ 
in Def.~\ref{mec:def} is convex on $[0,1]$. Furthermore, if $\beta>1$, then $\EC$ is also strictly convex on $[0,1]$.
\end{restatable}

Convexity is also advantageous from a computational perspective. For a convex function, every local minimum is a global minimum. Thus, in computing the minimum effort to consensus (MEC), it suffices to find any local minimum of $\EC(y)$ on $[0,1]$. When $\beta>1$, the function $\EC(y)$ is continuously differentiable and strictly convex. It therefore admits a unique optimal consensus point characterized by the first-order condition $d/dy\EC(y)=0$. Standard root-finding procedures, such as the bisection method applied to $d/dy\EC(y)$, locate the optimal consensus point within $\varepsilon$ accuracy in $O(\log(1/\varepsilon))$ iterations.

Nevertheless, for certain families of opinion distributions and values
of $\alpha,\beta$, the value of $\MEC$ can be computed directly, either
through a closed formula, through standard measures, or by exploiting
symmetries of the distribution. In the next section we explore these special cases.  

\subsection{Characterization and Special Cases for MEC}\label{section:special-cases}

Recall that $L^p$-dispersion of a probability distribution (with finite support) $\rho$ is given by

\begin{equation}
D_p(\rho)\defsymbol \min_{y\in[0,1]} \sum_{x\in\support(\rho)} \rho(x)\,|x-y|^p.
\end{equation}

The $L^p$-dispersion is a way of measuring how spread out a probabilistic distribution is using an $L^p$ type of distance. MEC can be characterized in terms of this measure as follows. 

\begin{restatable}[$L^\beta$-dispersion Characterization]{theorem}{mecLbeta}\label{thm:mec-Lbeta}
Let $\alpha \ge 1$ and $\beta \ge 1$, and let $\pi\in\Dset$ be an opinion distribution. Then
\(
\MEC[\alpha,\beta](\pi) = \Mass{\pi^\alpha} \cdot D_\beta(\overline {\pi^\alpha}).
\)
\end{restatable}

Thus, $\MEC$ equals the total transformed mass $\Mass{\pi^\alpha}$ times the $L^\beta$-dispersion of the normalized distribution $\overline{\pi^\alpha}$.

Two familiar cases arise as a corollary from the above theorem. When $\beta=1$, the effort function
is minimized at a median, and the resulting dispersion is the mean absolute
deviation from that median. When $\beta=2$, the minimizing point is the mean,
and the dispersion coincides with the variance. Thus, in these classical
settings, the minimal effort to consensus is simply the dispersion of
$\overline{\pi^\alpha}$, measured either by mean absolute deviation or by
variance, scaled by the total mass $\Mass{\pi^\alpha}$.

First recall that 
the \emph{mean absolute deviation about a median} of a finitely
supported probability distribution $\rho$ is defined as
\(
\mathrm{MAD}_{\mathrm{m}}(\rho)\defsymbol \sum_{x\in\support(\rho)} \rho(x)\,|x-m|,
\)
where $m$ is any median of $\rho$.
Also, the \emph{variance} of $\rho$ is given by
\(
\mathrm{Var}(\rho)=\sum_{x\in\support(\rho)} \rho(x)\,(x-\mu)^2,
\)
where $\mu$ is the mean of $\rho$.

\begin{restatable}[Median and Mean Consensus]{corollary}{specialCases}\label{cor:mec-special-cases}
Under the assumptions of Theorem~\ref{thm:mec-Lbeta}, the following hold:
\begin{enumerate}
\item 
\(
\MEC[\alpha,1](\pi)
=
\Mass{\pi^\alpha}\cdot \mathrm{MAD}_{\mathrm{m}}(\overline{\pi^\alpha})
=\EC[\alpha,1](m),
\)
where $m$ is a median of $\overline{\pi^\alpha}$.

\item 
\(
\MEC[\alpha,2](\pi)
=
\Mass{\pi^\alpha}\cdot \mathrm{Var}(\overline{\pi^\alpha})
=\EC[\alpha,2](\mu),
\)
where $\mu$ is the mean of $\overline{\pi^\alpha}$.
\end{enumerate}
\end{restatable}

\subsection{Symmetry}

We now consider opinion distributions whose symmetry determines their polarization.

\begin{definition}[Reflection and symmetry]
Let $\pi\in\Dset$ and let $c\in[0,1]$. We say that $c$ is an \emph{admissible reflection center} for $\pi$ if
\(
2c-x\in[0,1] \)
for all \(x\in\support(\pi).
\) 
If $c$ is an admissible reflection center for $\pi$, the
\emph{reflection of $\pi$ about $c$}, denoted $\pi^{\mathrm{ref}(c)}$,
is the opinion distribution with support
\(
\support(\pi^{\mathrm{ref}(c)})
=
\{\,2c-x \mid x\in\support(\pi)\,\}
\)
and masses
\(
\pi^{\mathrm{ref}(c)}(2c-x)\defsymbol \pi(x)
\) for all \( x\in\support(\pi).
\) 

We say that $\pi$ is \emph{symmetric about $c$} if $c$ is an admissible
reflection center for $\pi$ and
\(
\pi=\pi^{\mathrm{ref}(c)}.
\)
\end{definition}

The MEC measure is insensitive to whether masses lie to the left or to the
right of the opinion spectrum. What matters are the masses of the groups and
their distances from a candidate consensus point, not the orientation of the
configuration. Reflecting a distribution about a point $c$ simply exchanges the
positions of groups located on the left and on the right of $c$ while
preserving these distances. Consequently, the minimal effort required to reach
consensus remains unchanged.

For example, consider the two-mass distribution
\(
\pi=((40,60),(0.2,0.6)).
\)
Reflecting $\pi$ about $c=0.4$ yields
\(\
\pi^{\mathrm{ref}(c)}=((60,40),(0.2,0.6)).
\)
The two distributions only differ by reversing the positions of the groups
around $c$, hence the minimal
effort required to reach consensus is the same for both distributions. Indeed, $\MEC[2,2](\pi)=\MEC[2,2](\pi^{\mathrm{ref}(c)})\approx 177.2$.
The following lemma formalizes this invariance property. 

\begin{restatable}[Reflection preserves MEC]{lemma}{mecReflection}
\label{lem:mec-reflection}
Let $\pi\in\Dset$ and let $c$ be an admissible reflection center for $\pi$.
Then
\(
\MEC[\alpha,\beta](\pi^{\mathrm{ref}(c)})
=
\MEC[\alpha,\beta](\pi).
\)
\end{restatable}

Another operation of interest in the study of polarization consists in \emph{shifting} the entire distribution of opinions either to the left or to the right.   

\begin{definition}[Shifts]
Let 
\(
\pi=((\pi_1,\ldots,\pi_n),(x_1,\ldots,x_n))\in\Dset
\)
be an opinion distribution. A real number $t$ is an \emph{admissible shift} for $\pi$ if
\(
0 \le x_1+t \)
and \(
x_n+t \le 1.
\) If $t$ is admissible, the \emph{shift of $\pi$ by $t$} is the opinion distribution  
\(
\pi^{\mathrm{shift}(t)}
=
((\pi_1,\ldots,\pi_n),(x_1+t,\ldots,x_n+t)).
\) 
\end{definition}

When an opinion distribution is shifted along the spectrum, the relative
positions of the opinions remain unchanged; only their common location is
altered. For this reason, the level of polarization measured by MEC should
remain invariant under such shifts. 

Indeed, the preservation of MEC under shifts follows directly from its
invariance under reflection (Lem.~\ref{lem:mec-reflection}), owing to the identity
\(
\pi^{\mathrm{shift}(t)}
=
\left(\pi^{\mathrm{ref}(c)}\right)^{\mathrm{ref}(c+t/2)},
\)
which holds for admissible $c$ and $t$, and expresses a shift as the
composition of two reflections.

\begin{restatable}[Shift preserves MEC]{corollary}{mecShift}
\label{cor:mec-shift}
Let $\pi\in\Dset$ and let $t$ be an admissible shift for $\pi$. Then
\(
\MEC[\alpha,\beta](\pi^{\mathrm{shift}(t)})
=
\MEC[\alpha,\beta](\pi).
\)
\end{restatable}

We now consider opinion distributions that are symmetric about a point $c$. 
Such distributions represent situations in which opinions to the left and right of $c$ are balanced: for every group at distance $d$ from $c$ on one side, there is an equally large group at the same distance on the other. 
The effort function is therefore symmetric, and the minimal effort to consensus is attained at $c$.

\begin{restatable}[Symmetry]{lemma}{symmetryMec}
\label{lem:symmetry-mec}
Let $\alpha \ge 1$ and $\beta \ge 1$, and let $\pi\in\Dset$. If $\pi$ is symmetric about $c$, then
\(
\MEC[\alpha,\beta](\pi)
=
\EC(c).
\)
\end{restatable}

As an immediate consequence of the symmetry lemma (Lem.~\ref{lem:symmetry-mec}) we derive an algebraic expression for the polarization of extremal opinion distributions. The minimal effort to consensus of any extremal distribution can be achieved by moving its population to the middle point $1/2$ since every $\pi^*$ is symmetric about this point.   

\begin{restatable}{theorem}{mecPiStar}\label{th:mec-pi-star}
For every $\pi \in \Dset$, we have
\[
\MEC(\pi^*)
=
\EC[\alpha,\beta][\pi^*](1/2)
=
2\left(\frac{\Mass{\pi^*}}{2}\right)^\alpha
\left(\frac{1}{2}\right)^\beta .
\]
\end{restatable}

\subsection{Population Scaling}

The following theorem states a simple homogeneity property of the minimal effort to consensus: scaling all population masses by a factor $\lambda$ scales the value of $\MEC$ by $\lambda^{\alpha}$.

\begin{restatable}[Population Scaling]{theorem}{multByScalar}\label{th:multbyscalar}
Let $\pi=((\pi_1,\ldots,\pi_n),(x_1,\ldots,x_n))\in\Dset$ and let $\lambda>0$.
We define
\(
\lambda\pi \defsymbol ((\lambda\pi_1,\ldots,\lambda\pi_n),(x_1,\ldots,x_n)).
\)
Then
\(
\MEC(\lambda\pi)=\lambda^{\alpha}\MEC(\pi).
\)
\end{restatable}

As an immediate consequence of the previous theorem we obtain
Condition~H of Esteban and Ray \cite{esteban1994measurement}. 
A polarization measure $\Pol$ satisfies Condition~H if
\begin{equation}\label{eq:condition-H}    
\Pol(\pi)>\Pol(\pi') \Rightarrow \Pol(\lambda\pi)>\Pol(\lambda\pi')
\end{equation}

for every $\lambda>0,\pi,\pi' \in \Dset$. Condition~H expresses invariance to population scaling and is
standard in the theory of inequality measurement (see, e.g., \cite{foster85}). 

Another simple but useful application of the proposition is to obtain the polarization of a normalized distribution from that of the original distribution. This is especially useful if you want to compare populations of different sizes.  

\begin{restatable}{corollary}{conditionH}\label{cor:condition-H-normalization}
The following hold: 
\begin{enumerate}
\item $\MEC(\overline{\pi})=\frac{\MEC(\pi)}{(\Mass{\pi})^{\alpha}}.$
\item The $\MEC$ satisfies Condition~H.
\end{enumerate}
\end{restatable}

\section{Theoretical Applications}\label{sec:theoretical-applications}

\subsection{A Minority Principle}
Let us consider two-point opinion distributions 
supported on the extremes $\{0,1\}$ where one mass is smaller than the other, for example 
$\pi=((10,90),(0,1)).$ 
One might expect that the share of the total $(\alpha,\beta)$-effort of each mass required to reach an optimal consensus $y^*$ for $\pi$ should depend on the parameters $\alpha$
and $\beta$. In particular, it would seem natural that, by varying these
parameters, one could shift the burden of adjustment between the masses:
for some values, the larger mass might bear most of the effort, while
for others, the smaller mass might.

From the definition of $\EC$ the effort of the smaller mass to $y^*$ is $10^\alpha {y^*}^\beta$ while for the larger mass we obtain $90^\alpha (1-y^*)^\beta$. If $\beta=1$ the larger mass does not move at all, since from Cor.~\ref{cor:mec-special-cases}, we have $y^* = 1$. Thus, the smaller mass bears the whole effort. The situation is more subtle when $\beta>1$. In this case, the optimal
consensus point lies in the interior $(0,1)$, and thus both masses contribute to the
effort.  

Nevertheless, and somewhat surprisingly, it is actually the smaller mass that \emph{always} makes the greater effort to reach the optimal consensus, regardless of $\alpha$ and $\beta$. This is established below for any two-point distribution \(\pi\) with support \(\{l,r\}\), where \(l<r\).

First the following lemma confirms the intuition that the optimal consensus point lies closer to the larger mass.

\begin{restatable}{lemma}{yoptClosest}\label{lem:yopt-closest-to-largest-general}
Assume \(\beta>1\). Let \(0\le l<r\le 1\), and let \(\pi\in\Dset\) satisfy \(\pi(l)=A\) and \(\pi(r)=B\), with \(A,B>0\). Let \(y^*\) be such that \(\MEC(\pi)=\EC[\alpha,\beta][\pi](y^*)\). Then:
\begin{enumerate}
\item if \(y^*>\frac{l+r}{2}\), then \(B>A\);
\item if \(y^*<\frac{l+r}{2}\), then \(B<A\);
\item if \(y^*=\frac{l+r}{2}\), then \(B=A\).
\end{enumerate}
\end{restatable}

Having located the side of the optimal consensus point relative to the two masses, we now compare their respective contributions to the total effort.

\begin{restatable}{lemma}{greatestContribution}\label{lem:greatest-contribution-from-smallest-general}
Under the assumptions of Lem.~\ref{lem:yopt-closest-to-largest-general}, we have:
\begin{enumerate}
\item if \(y^*>\frac{l+r}{2}\), then \(A^\alpha (y^*-l)^\beta > B^\alpha (r-y^*)^\beta\);
\item if \(y^*<\frac{l+r}{2}\), then \(A^\alpha (y^*-l)^\beta < B^\alpha (r-y^*)^\beta\);
\item if \(y^*=\frac{l+r}{2}\), then \(A^\alpha (y^*-l)^\beta = B^\alpha (r-y^*)^\beta\).
\end{enumerate}
\end{restatable}

As an immediate consequence of the above lemmas and Cor.~\ref{cor:mec-special-cases}, we conclude that indeed it is the minority that always bears the greater effort.

\begin{restatable}[Minority Principle for Two-Point Distributions]{theorem}{minorityPrinciple}
\label{thm:minority-principle-general}
Let \(0\le l<r\le 1\), and let \(\pi\in\Dset\) satisfy \(\pi(l)=A\) and \(\pi(r)=B\), with \(A,B>0\) and \(A\neq B\). Let \(y^*\) be such that \(\MEC(\pi)=\EC[\alpha,\beta][\pi](y^*)\). Then:
\[
A<B \;\Longrightarrow\; A^\alpha (y^*-l)^\beta > B^\alpha (r-y^*)^\beta,
\]
and
\[
A>B \;\Longrightarrow\; A^\alpha (y^*-l)^\beta < B^\alpha (r-y^*)^\beta.
\]
\end{restatable}

\subsection{A Tipping Point Method}\label{sec:tipping-point} 

In Ex.~\ref{ex:group-identity}, we illustrated the role of group identity using a distribution $\pi$ and showing that transferring a small portion of mass at 0.75 toward the extreme right may in fact decrease polarization. Such a shift can reduce the effective weight of the original mass and thereby lower the overall level of polarization.

Nevertheless, if $\pi'$ is obtained from $\pi$ in Ex.~\ref{ex:group-identity} by transferring the entire mass
at $0.75$ to the extreme right, polarization does increase. Indeed, we have 
\(
\MEC[2,2](\pi)<\MEC[2,2](\pi') \approx 461.5 .
\)
This naturally suggests the following question: what is the minimal amount of
mass (\emph{critical mass}) that must move from the moderate right (0.75) to the extreme right (1.0) for
polarization to increase? In other words, at what point does such a shift begin
to raise polarization?

We now illustrate a simple method to solve the above tipping-point question for a family of two-point distributions with group identity and alienation parameters $\alpha=\beta=2.$ The phenomenon illustrates a broader fact: under MEC, polarization is not monotone in movement away from the optimal consensus point. Whether a partial transfer of mass away from this point raises or lowers polarization depends on how much mass is involved, and the Tipping Point method makes this dependence explicit and computable.

\begin{example}[A tipping point method for mass shift]
\label{ex:tipping-mass-shift-general}
Let \(m,n>0\) and \(0\le l<r<r'\le 1\). Consider the opinion distribution
\(
\pi=((m,n),(l,r)),
\)
and, for \(k\in(0,n]\), let
\(
\pi_{r:k\rightarrow r'}=((m,n-k,k),(l,r,r'))
\)
be the distribution obtained by moving mass \(k\) from \(r\) to the more extreme point \(r'\).

To study how polarization changes under this transfer, define
\[
\Delta(k)\defsymbol \MEC[2,2](\pi_{r:k\rightarrow r'})-\MEC[2,2](\pi).
\]

Set \(d\defsymbol r-l\), \(e\defsymbol r'-r\), and \(C\defsymbol m^2+n^2\). By Proposition~\ref{prop:mec-general-mass-shift-formulas}, we have
\[
\Delta(k)=\frac{k\,R_{m,n,d,e}(k)}{Q_{m,n}(k)},
\]
where
\[
R_{m,n,d,e}(k)
=
e^2Ck^3-2e^2nCk^2+\bigl(2d^2m^4+2de\,m^2C+e^2C^2\bigr)k-2d^2m^4n,
\]
and
\[
Q_{m,n}(k)=C(2k^2-2nk+C).
\]

Since \(Q_{m,n}(k)>0\) for all \(k\in[0,n]\), the sign of \(\Delta(k)\) for \(k>0\) is determined by \(R_{m,n,d,e}(k)\). Moreover,
\[
R_{m,n,d,e}(0)=-2d^2m^4n<0,
\qquad
R_{m,n,d,e}(n)=e(2d+e)m^2n(m^2+n^2)>0.
\]
By Proposition~\ref{prop:unique-root-tipping-general}, the polynomial \(R_{m,n,d,e}\) has a unique root \(k^\star\in(0,n)\). Therefore
\[
\MEC[2,2](\pi_{r:k^\star\rightarrow r'})=\MEC[2,2](\pi),
\]
and
\[
\MEC[2,2](\pi_{r:k\rightarrow r'})<\MEC[2,2](\pi)
\quad\text{for }0<k<k^\star,
\]
while
\[
\MEC[2,2](\pi_{r:k\rightarrow r'})>\MEC[2,2](\pi)
\quad\text{for }k^\star<k\le n.
\]

Thus \(k^\star\) acts as a tipping point: transferring a small amount of mass from \(r\) to \(r'\) initially reduces polarization, but once the transferred mass exceeds \(k^\star\), the polarization becomes strictly larger than that of the original distribution. \qed
\end{example}

\begin{remark}
By the Reflection Lemma (Lem.~\ref{lem:mec-reflection}), the analogous case in which a mass \(k\), with \(0<k\le m\), is moved leftward from \(l\) to a more extreme point \(l'<l\), reduces to the rightward case by reflection about \(c=(l'+r)/2\). Indeed, let
\[
\pi_{l:k\to l'}=((k,m-k,n),(l',l,r)).
\]
If \(r''\defsymbol l'+r-l\), then \(l'<r''<r\), \(\pi^{\mathrm{ref}(c)}=((n,m),(l',r''))\), and \((\pi_{l:k\to l'})^{\mathrm{ref}(c)}=((n,m-k,k),(l',r'',r))\). Thus the result follows from the rightward case.
\end{remark}

We now show an application of the tipping point method to the distribution in 
Ex.~\ref{ex:group-identity}.

\begin{example}
Let
\[
\pi=((40,60),(0.25,0.75)),
\]
so that
\[
\pi_{0.75:k\rightarrow1}=((40,60-k,k),(0.25,0.75,1)).
\]
Applying the analysis of Ex.~\ref{ex:tipping-mass-shift-general} with
\(m=40\), \(n=60\), \(l=0.25\), \(r=0.75\), and \(r'=1\), we obtain
\(d=\frac12\) and \(e=\frac14\). Thus the change in polarization is
determined by the cubic polynomial
\[
R_{40,60,\frac12,\frac14}(k)
=
325k^3-39000k^2+5050000k-76800000.
\]
This polynomial has a unique real root in \((0,60)\) 
\[
k^\star\approx 17.156.
\]
Therefore
\[
\MEC[2,2](\pi_{0.75:k\rightarrow1})<\MEC[2,2](\pi)
\quad\text{for }0<k<k^\star,
\]
while
\[
\MEC[2,2](\pi_{0.75:k\rightarrow1})>\MEC[2,2](\pi)
\quad\text{for }k^\star<k\le 60.
\]

Thus the polarization becomes strictly larger than that of \(\pi\) once the
transferred mass exceeds approximately \(17.156\).
\end{example}

We conclude by showing that the tipping point method above extends beyond the two-point setting. For any opinion distribution and any mass transferred from its maximal support point toward the extreme right (\(1\)), the resulting change in \(\MEC[2,2]\) is again a rational function whose sign is governed by a cubic polynomial.

\begin{restatable}[Mass transfer to the extreme]{proposition}{massTransfer}
\label{prop:mec22-mass-transfer}
Let
\(
\pi=((p_1,\dots,p_r),(x_1,\dots,x_r))
\)
be an opinion distribution with
\(
0\le x_1<\cdots<x_r<1
\)
and \(p_i>0\). For \(k\in[0,p_r]\), let \(\pi^{(k)}=\pi_{x_r:k\to 1}\) be the distribution obtained by moving mass \(k\) from the maximal support point \(x_r\) to \(1\). Then there exist polynomials \(P_\pi(k)\) and \(Q_\pi(k)\) such that
\[
\MEC[2,2](\pi^{(k)})-\MEC[2,2](\pi)
=
\frac{k\,P_\pi(k)}{Q_\pi(k)}
\qquad\text{for all }k\in[0,p_r],
\]
where \(P_\pi(k)\) is a cubic polynomial and \(Q_\pi(k)>0\) for all \(k\in[0,p_r]\). Consequently, for every \(k>0\) the sign of the change in polarization is determined by the cubic polynomial \(P_\pi(k)\).
\end{restatable}

\section{Shifting Away From Consensus}\label{sec:shifting-away}

This section is concerned with showing that \(\MEC\) is a polarization measure in the sense of Definition~\ref{pol:def}.

The first condition of Definition~\ref{pol:def} presents no difficulty. If \(\pi\) is a consensus distribution, then all its mass is concentrated at some opinion \(y\), and therefore \(\EC(y)=0\). It follows at once that \(\MEC(\pi)=0\).

The second condition is less immediate. We must show that \(\pi^*\) attains the greatest value of \(\MEC\) among all opinion distributions having total mass \(\Mass{\pi}\). This, however, is a more delicate matter and deserves some discussion.

Broadly speaking, the argument proceeds in two stages. We first show that any non-consensus distribution $\pi$ can be made \emph{strictly} more polarized by shifting mass away from consensus. By repeating this operation, one arrives at a distribution supported on $\{0,1\}$ or, in the exceptional case $\beta=1$, on $\{0,y^*,1\}$, where $y^*$ is an optimal consensus point of the final distribution. We then show that the extremal distribution $\pi^*$ is at least as polarized as any such final distribution.


There are, however, certain subtleties in the operation involved in the first stage. One may ask what happens when a mass, after being moved away from the optimal consensus point, merges with a mass already present. If \(\alpha\) is sufficiently large, the resulting mass may become so influential as to draw the new optimal consensus point very near to it. It then ceases to be obvious that the minimum effort required to reach consensus must increase.

\begin{example}
Let \(\alpha=\beta=2\), and consider \(\pi=((2,3,2),(0,0.4,1))\). The optimal consensus point is \( \frac{7.6}{17}\approx 0.447\) thus near the midpoint. Now move the mass at \(0.4\) leftward to \(0\), away from the optimal consensus point. It then merges with the mass already at \(0\), yielding \(\pi'=((5,2),(0,1))\). The new optimal consensus point becomes \(\frac{4}{29}\approx 0.138\). Thus a mass has been moved away from consensus, but in so doing it has formed a larger block, with one small on the other extreme, which draws the new consensus point sharply toward one extreme.
\end{example}

Nevertheless, it follows from the results below that the potential merging of masses involved in a shift away from consensus does not prevent polarization from increasing.

Also notice that Ex.~\ref{ex:group-identity} has already shown that moving a small \emph{portion} of a mass away from the optimal consensus point may in fact diminish polarization. For \(\alpha=\beta=2\), the optimal consensus point for \(\pi=((40,60),(0.25,0.75))\) is approximately \(0.59\), and yet, if one unit of mass is transferred from \(0.75\) to \(1\), the resulting distribution has smaller MEC. This may appear to conflict with the claim just made.

The conflict, however, is only apparent. The operation employed in the proof is not the arbitrary removal of a small portion of a mass, but the displacement of a \emph{whole} mass situated at a support point. The example is instructive precisely because it shows what may go wrong otherwise: the original group is reduced, while the detached fragment may remain too small to exert sufficient influence at its new position. In such a case, polarization need not increase.

Finally, there is a third observation, peculiar to the case \(\beta=1\). It arises when the consensus point is unique and the mass to be moved lies precisely at that point. By Cor.~\ref{cor:mec-special-cases}, the optimal point is then the median. If one moves the mass situated there, the median may be displaced with it. We shall refer to the resulting distribution as a \emph{median-shift}. In that case, as illustrated in the example below, the polarization need not increase.

First we need the following notions and notations.  

\begin{definition}[Shifting Masses]\label{def:movement}
Let \(\pi\in\Dset\) and let \(a,b\in[0,1]\). We use \(\pi_{a\to b}\) to denote the distribution obtained from \(\pi\) by shifting the mass \(\pi(a)\) from \(a\) to \(b\). Thus, \(\pi_{a\to b}(x)=\pi(x)\) for \(x\notin\{a,b\}\) and if $a \neq b$ then \(\pi_{a\to b}(b)=\pi(a)+\pi(b)\) and \(\pi_{a\to b}(a)=0\). If \(a=b\), we let \(\pi_{a\to a}\defsymbol \pi\).

Given $\pi$, we say that \(\pi_{a\to b}\) is a (distribution) \emph{shift away from consensus} for given \(\alpha,\beta\ge 1\) if \(b<a\le y^*\) or \(y^*\le a<b\) for every $(\alpha,\beta)$-optimal consensus point \(y^*\) of \(\pi\).

Finally, \(\pi_{a\to b}\) is said to be a \emph{median-shift} if \(\beta=1\), \(b\neq a=y^*\), \(\pi(a)>0\), and \(y^*\) is the only $(\alpha,\beta)$-optimal consensus point of \(\pi\).
\qed 
\end{definition}

\begin{example}
Assume \(\alpha=\beta=1\). Let \(\pi=((20,60,20),(0,0.5,1))\) and \(\pi_{0.5\to 0.9}=((20,60,20),(0,0.9,1))\). Since \(0.5\) is the only optimal consensus for \(\pi\) (i.e., its only median), the distribution \(\pi_{0.5\to 0.9}\) is a median-shift of \(\pi\). Furthermore, the polarization remains unchanged by this shift away from consensus: i.e., \(\MEC(\pi_{0.5\to 0.9})=\MEC(\pi)\).
\end{example}

We now state the main theorems needed for proving the second condition for $\MEC$ being a polarization measure. 

The following theorem establishes that shifting away from consensus increases polarization, which is the first step of our argument.

\begin{restatable}[Shifting Away from Consensus]{theorem}{shiftingAway}
\label{thm:mec-increases:body}
Let $\pi \in \Dset$ and $a\in\support(\pi)$. Suppose that  $\pi_{a\to b}$ is a shift away from consensus and is not a median-shift. Then $\MEC(\pi_{a\to b})> \MEC(\pi)$.
\end{restatable}

The case of a median-shift is explicitly excluded. The proof divides into the cases $\beta=1$ and $\beta>1$. If $\beta=1$, we show that $\pi$ and $\pi_{a\to b}$ have a common median $m$. By Cor.~\ref{cor:mec-special-cases}, $\MEC(\pi_{a\to b})-\MEC(\pi)=\EC[\alpha,\beta][\pi_{a\to b}](m)-\EC[\alpha,\beta][\pi](m)$, and this is easily seen to be strictly positive. If $\beta>1$, let $y^*$ be the unique optimal consensus point of $\pi$. We consider the function $F(t)\defsymbol \MEC(\pi_{a\to t})$, where $t$ ranges over points lying farther from $y^*$ than $a$ on the same side of $y^*$. We then prove that $F'(t)>0$. It follows that $F(b)>F(a)=\MEC(\pi)$, hence $\MEC(\pi_{a\to b})>\MEC(\pi)$.

The next theorem gives the second part of our strategy by reflecting a simple principle: when all opinion is concentrated at the two extremes, polarization is greatest when the two opposing groups are equal.

\begin{restatable}[Two-Point Polarization]{theorem}{twoPointPolarization} \label{thm:two_massed_distribution:body}  For $A >0$, 
let \(\pi_A\in\Dset\) satisfy
\(\pi_A(0)=A\) and \(\pi_A(1)=\Mass{\pi_A}-A\). Then  $\MEC(\pi_A) \leq  \MEC(\pi_A^*)$.
\end{restatable}
 
Intuitively, as $A$ tends to $0$, the distribution $\pi_A$ approaches the consensus distribution concentrated at the extreme right, whose polarization is $0$. The proof of the theorem proceeds by showing that the minimum effort required for consensus increases as the imbalance between the two extreme masses diminishes. It follows that this effort is maximal when the two masses are equal, that is, at the balanced distribution $\pi_A^*$.

More precisely, the proof reduces the question to establishing the maximum value of the function $g(A)\defsymbol \MEC(\pi_A)$ for $A\in[0,1]$. Since $\MEC$ satisfies Condition~H (Cor.~\ref{cor:condition-H-normalization}), it is enough to consider distributions $\pi_A$ with total mass $\Mass{\pi_A}=1$. By the Reflection Lemma (Lem.~\ref{lem:mec-reflection}), we have $g(A)=\MEC(\pi_A)=\MEC(\pi_A^{\mathrm{ref}(1/2)})=g(1-A)$; thus $g$ is symmetric about $1/2$. The remainder of the proof shows that $g$ increases on $[0,1/2]$. It then follows that $g$ attains its maximum at $A=1/2$. Hence $\pi_A^*$ maximizes $\MEC$ among two-point distributions supported at the extremes.

Finally, we arrive to the main theorem.
\begin{restatable}{theorem}{mecIsPolarization}\label{th:mec-is-a-polarization-measure} The function $\MEC$ is a polarization measure (Def.~\ref{pol:def}).
\end{restatable}
As mentioned above, the first condition of Def.~\ref{pol:def} is immediate. For the second, namely $\MEC(\pi)\leq \MEC(\pi^*)$, we argue by repeatedly applying Th.~\ref{thm:mec-increases:body}, each time performing a non-median shift away from consensus. If $\beta>1$, this process continues until one reaches a distribution $\pi'$ more polarized than $\pi$, with $\Mass{\pi'}=\Mass{\pi}$ and $\support(\pi')=\{0,1\}$. The conclusion then follows from Th.~\ref{thm:two_massed_distribution:body}.

If $\beta=1$, the same argument applies, except that the final distribution $\pi'$ may also have support $\{0,y^*,1\}$, where $0<y^*<1$ is the unique median, and therefore the unique optimal consensus point, of $\pi'$. At that stage no further non-median shift away from consensus is possible.

However, since $y^*$ is the unique median of $\pi'$, we have $\pi'(0)<\Mass{\pi}/2=\pi^*(0)$ and $\pi'(1)<\Mass{\pi}/2=\pi^*(1)$. Moreover, every point of $[0,1]$ is a median of $\pi^*$, and hence $y^*$ is an optimal consensus point for $\pi^*$. It follows that $\MEC(\pi')=\pi'(0)^\alpha y^*+\pi'(1)^\alpha(1-y^*)<\pi^*(0)^\alpha y^*+\pi^*(1)^\alpha(1-y^*)=\MEC(\pi^*)$. This completes the proof.

\section{Polarization by Effective Antagonism}\label{sec:er-comparison}

In their seminal paper \cite{esteban1994measurement}, Esteban and Ray laid down three axioms that a measure of polarization should satisfy. They also introduced a parametric measure, which we denote by \ERLabel, and stated that it is the only measure of certain form that satisfies these axioms.

\subsection{The Esteban--Ray Measure}

The Esteban--Ray measure (\ERLabel) is essentially pairwise. It adds up the tension between every pair of groups, weighting opinion difference by group size and by a term emphasizing group identity.

\begin{definition}[Esteban--Ray measure \cite{esteban1994measurement}]\label{ER:measure} The Esteban--Ray measure \(\ERLabel:\Dset \to \PositiveReals\) is given by
\begin{equation}
\ER(\pi) \defsymbol K \sum_{x\in\support(\pi)}\sum_{y\in\support(\pi)}\pi(x)^{1+\alpha}\pi(y)|x - y|,
\end{equation}
where $K > 0$ is a normalization constant, and $\alpha \in (0, \alpha^*]$, with $\alpha^* \approx 1.6$ empirically derived.
\end{definition}
From the results \cite{esteban1994measurement} it follows that \ERLabel{} is a polarization measure in the sense of Def.~\ref{pol:def}. The guiding idea in Def.~\ref{ER:measure} is that polarization grows when there are large groups that are far apart. The $\MEC$ measure also captures this but from a different perspective: It asks what is the least effort needed to bring the whole distribution to a consensus.

The optimization-based character of $\MEC$ yields information that \ERLabel{} does not provide, namely, the optimal consensus points. This is important for two reasons. On the one hand, it supplies a geometric structure that is useful in proving properties such as the increase of polarization under shifts away from optimal consensus (Th.~\ref{thm:mec-increases:body}) or that minorities always bear the greater effort (Th.~\ref{thm:minority-principle-general}). On the other hand, it has a direct social significance, since it indicates the points at which consensus can be achieved with least effort in a given population.

Furthermore, like \ERLabel, $\MEC$ uses the parameter $\alpha$ to reflect group identity. It differs from \ERLabel, however, in introducing a further parameter, $\beta$, which governs how severely distance from consensus is penalized. This not only allows $\MEC$ to distinguish more sharply between moderate and extreme deviations, but also determines the structure of the set of optimal consensus points, since for $\beta>1$ the optimal consensus point is unique, whereas for $\beta=1$ it need not be.

The two measures agree on an important intuition: societies split into two large and distant blocs should be highly polarized. In particular, the symmetric two-extreme distribution $\pi^*$ gives  maximal polarization under both approaches. But they justify this in different ways: \ERLabel{} sees two large hostile camps while $\MEC$ sees the greatest resistance to consensus.

\subsection{The \texorpdfstring{\ERLabel}{ER} Axioms}\label{sec:er-axioms}

Despite proceeding from a different conception of polarization, $\MEC$ satisfies all three of Esteban and Ray's axioms throughout its standing parameter range $\alpha,\beta\geq 1$. The only restriction is that Axiom~1 requires the strict inequality $\alpha>1$. This is not specific to $\MEC$: under the parametrization of \cite{esteban1994measurement}, the mass exponent is $1+\alpha_{\ERLabel}$ with $\alpha_{\ERLabel}>0$, which is the same condition (mass exponent strictly greater than $1$) expressed in a different variable. The strict inequality is, in this sense, intrinsic to the identification dimension of polarization rather than an accident of either formulation. The need for it is also natural in light of Cor.~\ref{cor:mec-special-cases}: when $\alpha=1$, $\MEC$ reduces to a standard dispersion measure (Earth Mover's Distance, mean absolute deviation about a median, or variance), and dispersion measures are known not to capture the consolidation phenomenon described by Axiom~1.

The axioms stated below distill the conceptual insights of Esteban and Ray into precise, testable properties governing how the consolidation of proximate clusters, increased separation among groups, and redistribution of mass toward the extremes should affect a polarization index. They thus serve as normative criteria: any measure that satisfies them faithfully captures the core sociological intuitions about group cohesion and antagonism \cite{kawada2018characterization}.

\paragraph{Axiom 1 (Increased identification).}
The first \ERLabel{} axiom says, in effect, that polarization increases when two groups of size $q$, sufficiently small relative to a larger third group of size $p$ and sufficiently close to one another, are fused into one.

\begin{axiom}[Polarization via increased identification]\label{axiom:1}
A polarization measure $\Pol:\Dset\to \PositiveReals$ satisfies \ERLabel{} Axiom~1 if it satisfies the \emph{increased identification condition}: For any $p>0$ and any $x>0$, there exist $\varepsilon>0$ and $\mu>0$ such that, for all $y>x$ and all $q<p$, if $y-x<\varepsilon$ and $0<q<\mu p$, then
\[
\Pol((p,q,q),(0,x,y)) < \Pol\!\left((p,2q),\left(0,\frac{x+y}{2}\right)\right).
\]
\begin{figure}[H]
    \centering
\includegraphics[height=0.16\textheight,keepaspectratio]{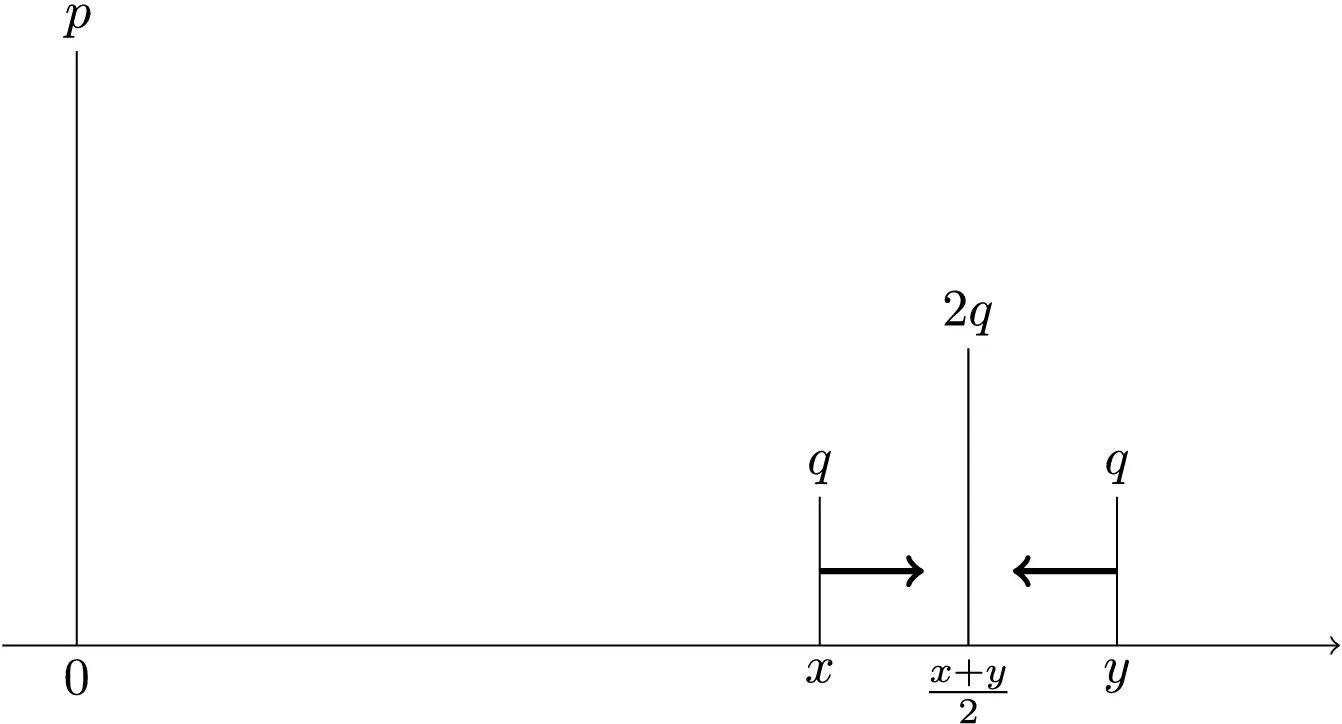}
    \caption{Esteban--Ray Axiom 1.}
    \label{fig:eraxiom1}
\end{figure}
\end{axiom}

\noindent Intuitively, consolidating smaller, closely positioned groups into fewer, larger groups enhances identification and thus polarization.

It turns out that the condition $\alpha>1$ is not merely sufficient, but also necessary, for our $\MEC$ functions to satisfy Axiom~1.

\begin{restatable}{theorem}{mecAxiomOne}\label{th:mec-axiom-1} The polarization measure $\MEC$ satisfies Axiom~1 if and only if $\alpha > 1$.
\end{restatable}

An immediate consequence of this theorem is that the instances $\MEC[1,1]$ and $\MEC[1,2]$ do not satisfy Axiom~1. In view of Th.~\ref{th:mec-as-emd} and Cor.~\ref{cor:mec-special-cases}, these correspond to well-known probabilistic measures, namely Earth Mover's Distance, the mean absolute deviation about a median, and the variance.

The proof of Th.~\ref{th:mec-axiom-1} compares the effort functions of 
\(\pi=((p,q,q),(0,x,y))\) and the merged distribution
\(\tilde\pi=((p,2q),(0,(x+y)/2))\). Writing \(h=y-x\), the difference of these functions at a
point \(z\in(0,x)\) can be written as \(q^\alpha\Delta(A,h)\), where
\(A=x-z\). The quantity \(\Delta(A,h)\) splits into two parts: a non-negative
term coming from the convexity of \(t\mapsto t^\beta\), and a negative term
proportional to \(2^\alpha-2\), which measures the gain in identification
produced by merging the two masses \(q\) and \(q\) into a single mass \(2q\).

The smallness of the mass \(q\) with respect to $p$ ensures that the relevant optimal consensus points lie near the mass \(p\) at \(0\), so that the
comparison is made in the right region. When \(\beta>1\), the smallness of the distance
\(h\) ensures that the positive convexity term remains controlled. If
\(\alpha>1\), the negative identification term is present and, for small enough
\(h\), it dominates; hence the merged distribution has larger MEC. If
\(\alpha=1\), this term disappears. One then gets equality when \(\beta=1\),
and the reverse inequality when \(\beta>1\). Thus the threshold is exactly
\(\alpha>1\).

\paragraph{Axiom 2 (Increased alienation).}
The second \ERLabel{} axiom says that polarization increases when one group moves towards another, more distant group, thereby widening the separation between groups at opposing extremes.

\begin{axiom}[Polarization via increased alienation]\label{axiom:2}
We say a polarization measure $\Pol:\Dset\to\PositiveReals$ \emph{satisfies \ERLabel{} Axiom~2} if, for every $p,q,r>0$ with $p>r$, every $x,y>0$ with $\lvert y-x\rvert<x<y$, and every $\Delta\in(0,y-x)$, the inequality
\[
\Pol\bigl((p,q,r),(0,x,y)\bigr)<\Pol\bigl((p,q,r),(0,x+\Delta,y)\bigr)
\]
holds.
\end{axiom}

\begin{figure}[H]
    \centering
    \includegraphics[height=0.16\textheight,keepaspectratio]{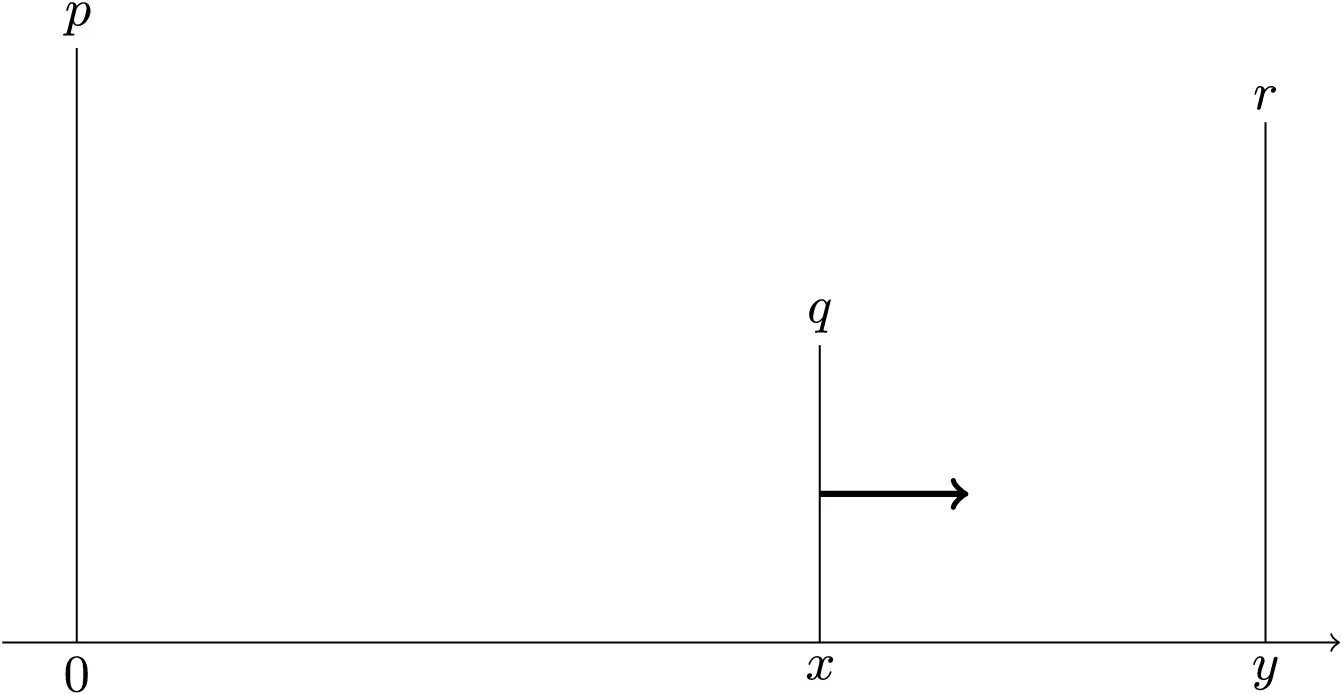}
    \caption{Esteban--Ray Axiom 2.}
    \label{fig:eraxiom2}
\end{figure}

\noindent This axiom captures the increased antagonism resulting from clearer separation among groups positioned at opposing extremes.

\begin{restatable}{theorem}{mecAxiomTwo}\label{th:mec-axiom-2}
The polarization measure $\MEC$ satisfies Axiom~2.
\end{restatable}

When $\beta>1$, the hypothesis $p>r$ together with $x>y-x$ forces the unique optimal consensus point $y^*$ of $\pi=((p,q,r),(0,x,y))$ to lie strictly to the left of $x$. The displacement of the middle mass from $x$ to $x+\Delta$ is therefore a shift away from consensus in the sense of Def.~\ref{def:movement}, and Th.~\ref{thm:mec-increases:body} applies.

The case $\beta=1$ requires a small additional argument. Here $x$ may itself be the unique optimal consensus point of $\pi$, in which case the shift $\pi_{x\to x+\Delta}$ is a median-shift, a situation explicitly excluded by Th.~\ref{thm:mec-increases:body}. The asymmetry $p>r$, however, is enough to ensure that even this median-shift strictly increases polarization, since by Cor.~\ref{cor:mec-special-cases} the difference reduces to $(p^\alpha-r^\alpha)\,\Delta>0$. The remaining sub-cases are non-median shifts away from consensus and so are handled by Th.~\ref{thm:mec-increases:body} as before. The full argument is given in Appendix~\ref{appendix:proof-mec-axiom-2}.

\paragraph{Axiom 3 (Dispersion to extremes).}
The third \ERLabel{} axiom says that polarization increases when central mass is redistributed symmetrically toward the two extremes.

\begin{axiom}[Polarization via dispersion to extremes]\label{axiom:3}
A polarization measure $\Pol:\Dset\to\PositiveReals$ satisfies \ERLabel{} Axiom~3 if for any $p,q>0$, any $x,y>0$ with $x=y-x$, and any $\Delta\in(0,q/2)$,
\[
\Pol((p,q,p),(0,x,y)) < \Pol((p+\Delta,q-2\Delta,p+\Delta),(0,x,y)).
\]
\begin{figure}[H]
    \centering
    \includegraphics[height=0.16\textheight,keepaspectratio]{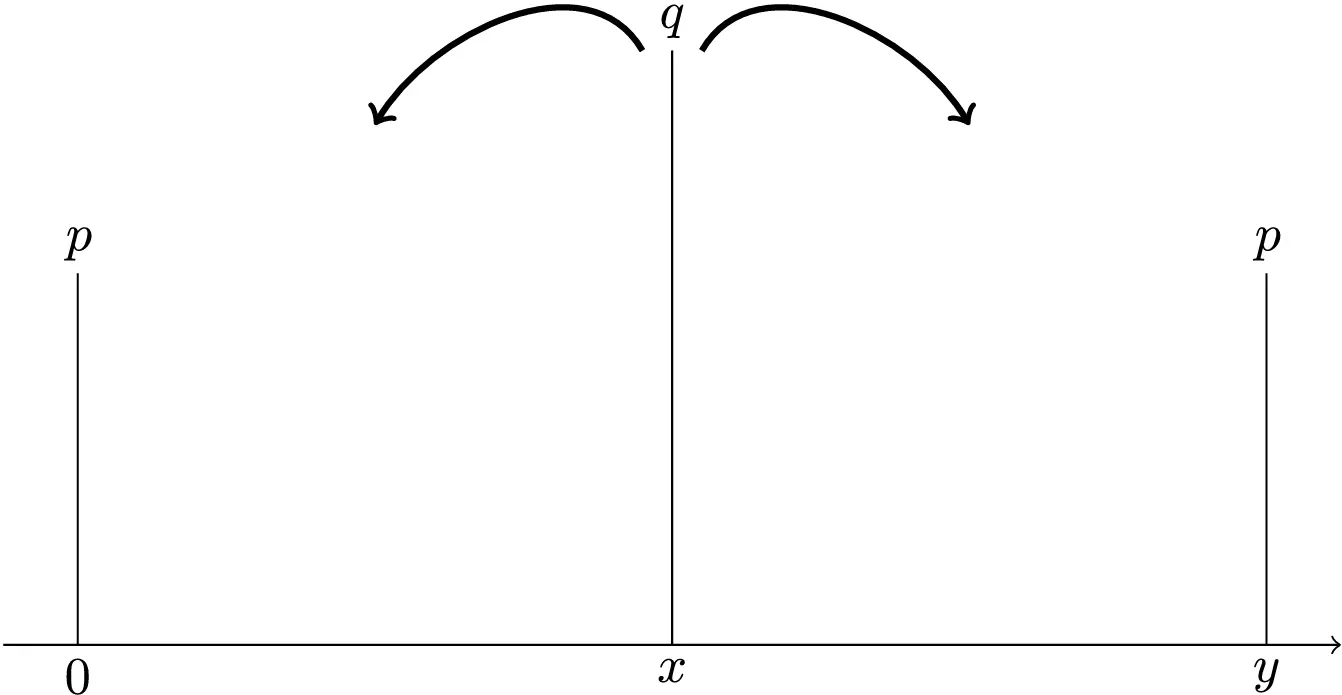}
    \caption{Esteban--Ray Axiom 3.}
    \label{fig:eraxiom3}
\end{figure}
\end{axiom}

\noindent This axiom reflects the intuitive notion that diminishing central positions in favor of extremes increases polarization.

\begin{restatable}{theorem}{mecAxiomThree}\label{th:mec-axiom-3}
The polarization measure $\MEC$ satisfies Axiom~3.
\end{restatable}

The argument exploits the symmetry of both distributions about $x$: by Lem.~\ref{lem:symmetry-mec}, the optimal consensus point is $x$ for both, and direct evaluation of the effort at $x$ shows the inequality.

\subsection{Central-Split Monotonicity: A New Property}\label{sec:property-1}

Beyond the three foundational axioms of Esteban and Ray, we introduce a complementary property derived as an intermediate refinement of Axiom~3.

\begin{axiom}[Central-split monotonicity, Property~1]\label{axiom:property-1}
Let \(P:\Dset\to\PositiveReals\) be a polarization measure (Def.~\ref{pol:def}). We say that \(P\) \emph{satisfies Property~1} if for any \(p,q \ge 0\), any positions \(x,y \ge 0\) and \(z>0\) satisfying \(x = y - x\) and \(x \ge z\), and any \(\Delta \in (0, q/2)\):
\[
\Pol((p,q,p), (0,x,y)) < \Pol((p, \Delta, q - 2\Delta, \Delta, p), (0,x - z, x, x + z, y)).
\]
\begin{figure}[H]
    \centering
    \includegraphics[height=0.16\textheight,keepaspectratio]{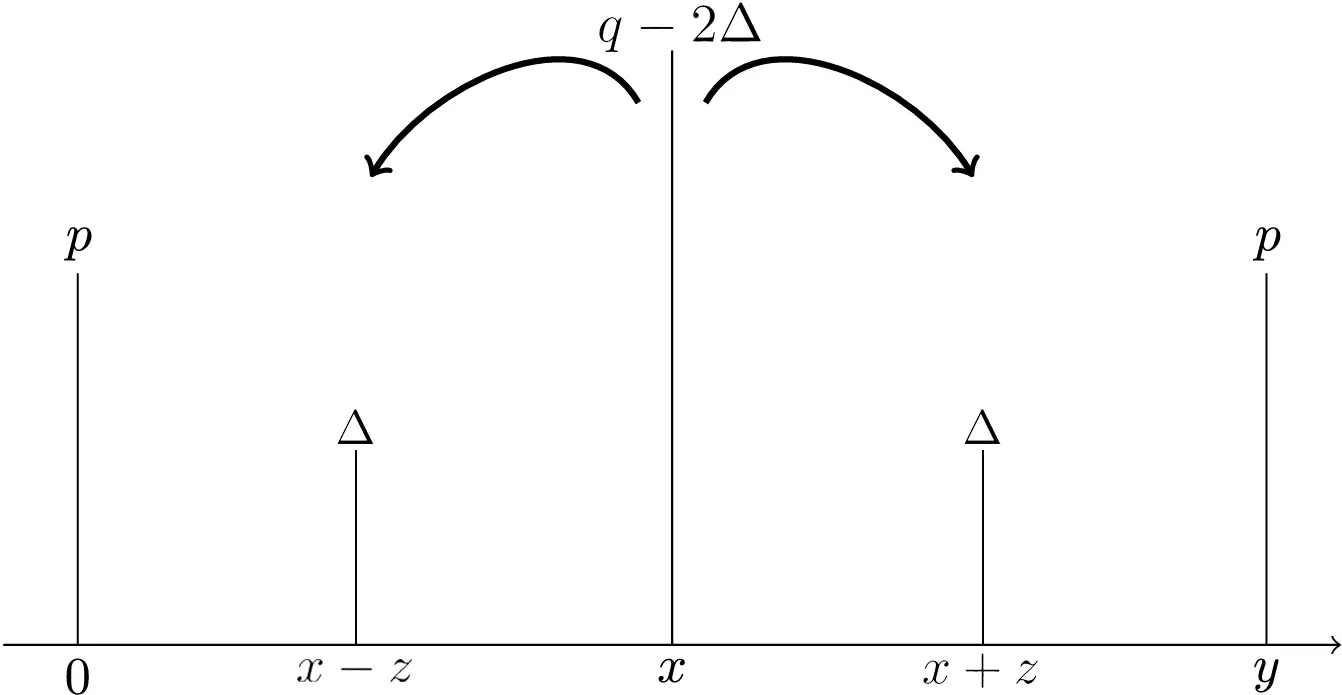}
    \caption{Property 1.}
    \label{fig:property1}
\end{figure}
\end{axiom}

\noindent Intuitively, this property indicates that polarization increases when a single central group splits into smaller groups positioned symmetrically around the original midpoint. The resulting distribution captures emerging nuances in polarization, reflecting subtle but meaningful shifts in inter-group antagonism and identification that the original Axiom~3 (which reallocates central mass directly to the extremes) does not detect.

The property is not implied by the Esteban--Ray axioms. For instance, the standard parametrization \ERLabel[0.8]{} fails it on the central-split instance obtained from
\[
((0.25,0.5,0.25),(0,0.5,1))
\]
by taking \(\Delta=0.1\) and \(z=0.1\), which gives
\[
((0.25,0.1,0.3,0.1,0.25),(0,0.4,0.5,0.6,1)).
\]
The \ERLabel[0.8]{} score decreases from approximately \(0.53\) to \(0.44\), contrary to the monotonicity required by Property~1.

\begin{restatable}{theorem}{mecPropertyOne}\label{th:mec-property-1}
For $\alpha\geq 1$ and $\beta\geq 1$, the polarization measure $\MEC$ satisfies Property~1.
\end{restatable}

The argument is similar to that of Axiom~3: by symmetry the optimal consensus point of both distributions is $x$, and direct evaluation yields the inequality.

\section{Experimental Validation}\label{sec:experiments}

We complement the theoretical analysis with an empirical evaluation that places $\MEC$ alongside the principal alternative measures in the literature. The evaluation has two parts: a pairwise concordance analysis on an exhaustive corpus of opinion distributions, and a validation against expert judgments using the benchmark of Koudenburg, Kiers, and Kashima~\cite{koudenburg2021new}.

\subsection{Candidate Measures and Computational Universe}\label{sec:candidates-corpus}

\paragraph{Candidate measures.}
We compare $\MEC$ against five alternatives. The Esteban--Ray measure (\ERLabel), defined in Def.~\ref{ER:measure}, formalizes polarization through intra-group identification and inter-group alienation. The Earth Mover's Distance (\EMDLabel) corresponds to the transport interpretation of \(\MEC[1,1]\). The Van der Eijk index (\VdELabel) is tailored for ordered rating scales, conceptualizing polarization as the complement of agreement, which is quantified through a unimodality-based index. The Tastle--Wierman Dissention measure (\TWDLabel) is an ordinal dispersion index derived from information-theoretic ideas, combining category probabilities with their normalized distance from the mean opinion. Finally, the Koudenburg--Kiers--Kashima Index (\KILabel) refers to the empirical measure derived by Koudenburg, Kiers, and Kashima~\cite{koudenburg2021new} by regressing polarization ratings provided by experts; \KILabel{} is used below as a compact label for this three-author index. This method exhibits theoretical limitations, particularly in distributions with significant clustering at the extremes. These measures, together with $\MEC$, constitute the candidate set we evaluate.

\paragraph{Computational universe.}
We define the corpus of opinion distributions employed in our analyses as the set of integer-valued opinion distributions of $100$ agents over $n=5$ ordinal positions (a standard Likert scale). Formally, let $\Dset_{100}^{(5)} = \big\{ \pi=(\pi_1,\dots,\pi_5)\in\mathbb{N}^5 : \Mass{\pi} = 100 \big\}$. The cardinality of this space is given by the stars-and-bars formula:
\[
\big|\Dset_{100}^{(5)}\big| = \binom{100+5-1}{5-1} = 4{,}598{,}126.
\]
However, every candidate measure $\Pol$ considered in this work satisfies the reflection symmetry $\Pol(\pi) = \Pol(\pi^{\text{rev}})$, where $\pi^{\text{rev}} = (\pi_5, \pi_4, \pi_3, \pi_2, \pi_1)$. This allows us to reduce the search space by selecting, for each pair $\set{\pi, \pi^{\text{rev}}}$, a single canonical representative (the lexicographically smaller of the two). Let $\mathscr{S} \subset \Dset_{100}^{(5)}$ denote the set of self-symmetric (palindromic) distributions satisfying $\pi = \pi^{\text{rev}}$, i.e., $\pi_1 = \pi_5$ and $\pi_2 = \pi_4$. The remaining $\big|\Dset_{100}^{(5)}\big| - |\mathscr{S}|$ distributions pair off into distinct mirror pairs, so the number of unique canonical representatives is:
\[
\big|\widetilde{\Dset}_{100}^{(5)}\big| = |\mathscr{S}| + \frac{\big|\Dset_{100}^{(5)}\big| - |\mathscr{S}|}{2} = \frac{\big|\Dset_{100}^{(5)}\big| + |\mathscr{S}|}{2}.
\]
The palindromic count $|\mathscr{S}|$ equals the number of non-negative integer solutions to $2(\pi_1 + \pi_2) + \pi_3 = 100$, which is $\binom{52}{2} = 1{,}326$, yielding:
\[
\big|\widetilde{\Dset}_{100}^{(5)}\big| = \frac{4{,}598{,}126 + 1{,}326}{2} = 2{,}299{,}726.
\]

\paragraph{Choice of parameters for empirical analyses.}
The theoretical results in this paper are established for the full parametric
family $\MEC$. However, the empirical analyses require fixing a concrete representative. We use $\MEC[2,1.15]$ as the
reference parametrization. This choice is supported by an exhaustive search on the expert-judgment
benchmark of Koudenburg, Kiers, and Kashima: within the MEC family, no parameter pair was found to strictly outperform
$(2,1.15)$ in Kendall's $\tau$ against the expert ranking. Rather than a unique
optimum, the search revealed a set of equivalent solutions, among which
$(2,1.15)$ is the simplest and historically established representative.

\subsection{Pairwise Concordance via Kendall's \texorpdfstring{$\tau_B$}{tau-B}}\label{sec:pairwise-concordance}

A polarization measure is used ordinally: any strictly increasing transformation of it gives an equivalent measure for comparing or ranking distributions. The agreement between two measures should therefore be assessed by an index invariant under such transformations, which excludes Pearson's $r$ (invariant only under affine transformations) and motivates a rank correlation. We use Kendall's $\tau_B$, which equals, up to a standard tie correction, the fraction of pairs of distributions that two measures rank concordantly minus the fraction they rank discordantly. The $B$-variant adjusts for the tied ranks that arise when several distributions in the corpus receive the same score under a given measure.

For each distribution in the computational universe we computed polarization scores under all candidate measures and then pairwise $\tau_B$ across the corpus. In the resulting correlation matrix, the row labelled \EMDLabel{} coincides with $\MEC[1,1]$, by the identity $\MEC[1,1](\pi)=\min_{y\in[0,1]}\EMD(\pi,\pi_y)$ proved in Appendix~\ref{appendix:emd-definition}.

\begin{figure}[H]
    \centering
    \includegraphics[width=0.78\linewidth]{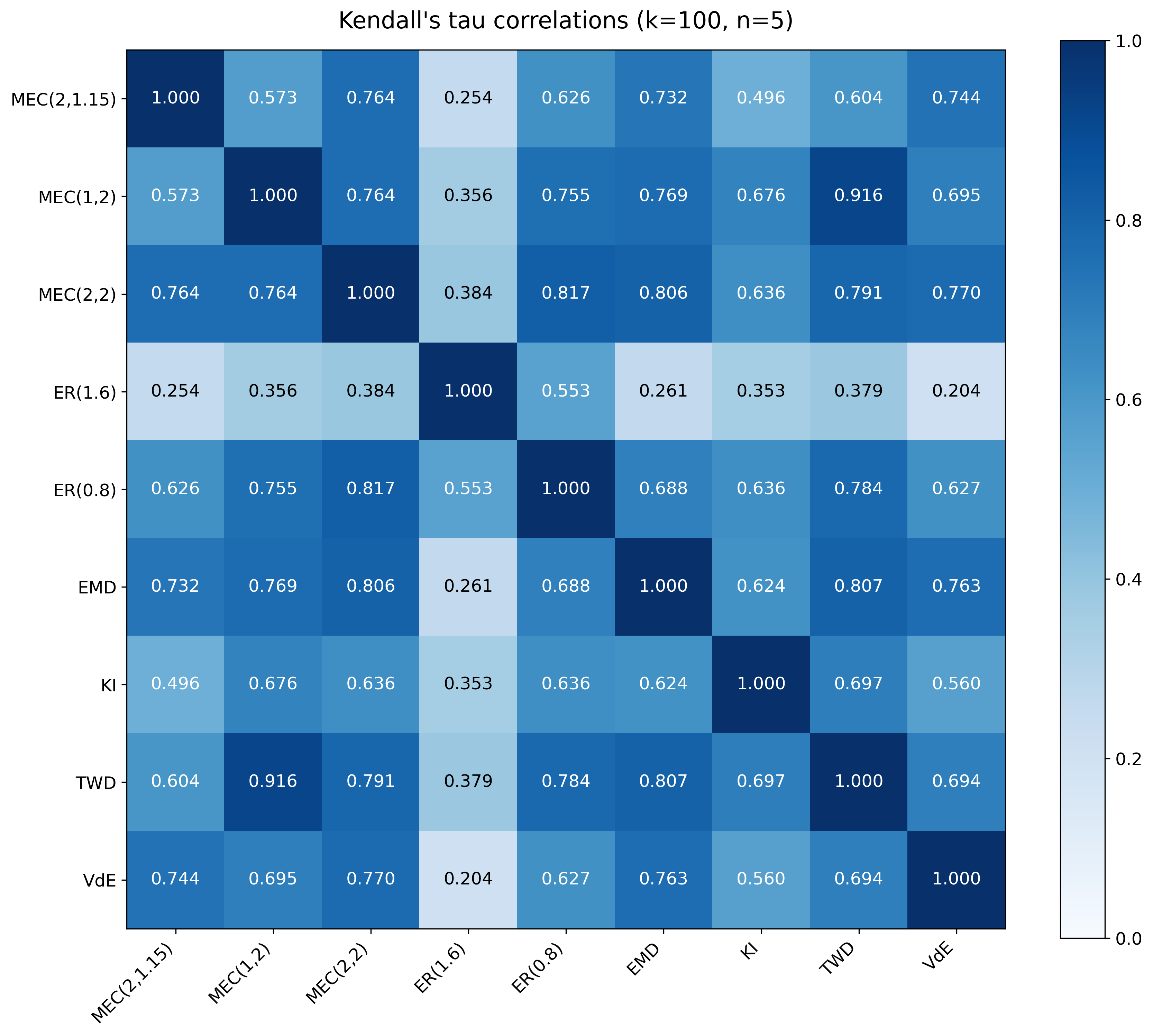}
    \caption{Kendall's $\tau$ correlations among candidate polarization measures.}
    \label{fig:kendall-tau-correlation}
\end{figure}

Several interesting patterns emerge in the correlations. \TWDLabel{} shows particularly high agreement with $\MEC[1,2]$ ($\tau_B = 0.916$). This is consistent with their underlying structures: by Cor.~\ref{cor:mec-special-cases}, $\MEC[1,2]$ is the variance of $\bar\pi$ scaled by total mass, and \TWDLabel{} aggregates category probabilities weighted by distance from the mean opinion. Both are therefore mean-based dispersion measures, and the close empirical agreement reflects this shared structure rather than a coincidence between unrelated constructions. Conversely, \ERLabel[1.6]{} displays weak alignment with other indices ($\tau_B < 0.4$ in most cases), reflecting its distinctive emphasis on large-group identification effects under high $\alpha$. The value $\alpha=1.6$ is the upper bound $\alpha^*$ of the range $(0,\alpha^*]$ that Esteban and Ray~\cite{esteban1994measurement} prove admissible under their axioms (see Section~\ref{sec:er-axioms}). They note that larger $\alpha$ corresponds to greater \emph{polarization sensitivity}, i.e., a larger departure from inequality measurement.

\KILabel{} shows moderate correlations with most other indices (typically $0.5<\tau_B<0.7$), with the highest alignment to \TWDLabel{} ($\tau_B=0.697$). However, \KILabel{} has a significant theoretical limitation. By construction, it returns zero polarization whenever the population is absent from one side of the scale, i.e.\ whenever $\pi_1=\pi_2=0$ or $\pi_4=\pi_5=0$, even though distributions such as $(0,0,0.1,0.4,0.5)$ which clearly exhibits a non-consensus  structure. This illustrates the broader limitations of direct linear interpolation in capturing the complex structure of opinion polarization.

\subsection{Validation Against Expert Judgments}

Beyond the all-distribution analysis, we validated each polarization measure against expert judgment by replicating the experimental protocol introduced by Koudenburg, Kiers, and Kashima~\cite{koudenburg2021new}, where 60 researchers in polarization studies were asked to evaluate 15 hypothetical histograms, assigning each a score from 0 (not polarized) to 100 (maximally polarized). This ranking served as a human-derived baseline to evaluate how well each computational measure captures expert intuition. To quantify the alignment between each measure and expert judgment, we computed Kendall's $\tau_B$ between the expert ranking and the ordering induced by each measure across the 15 distributions.

\begin{figure}[H]
    \centering
    \includegraphics[width=0.6\linewidth]{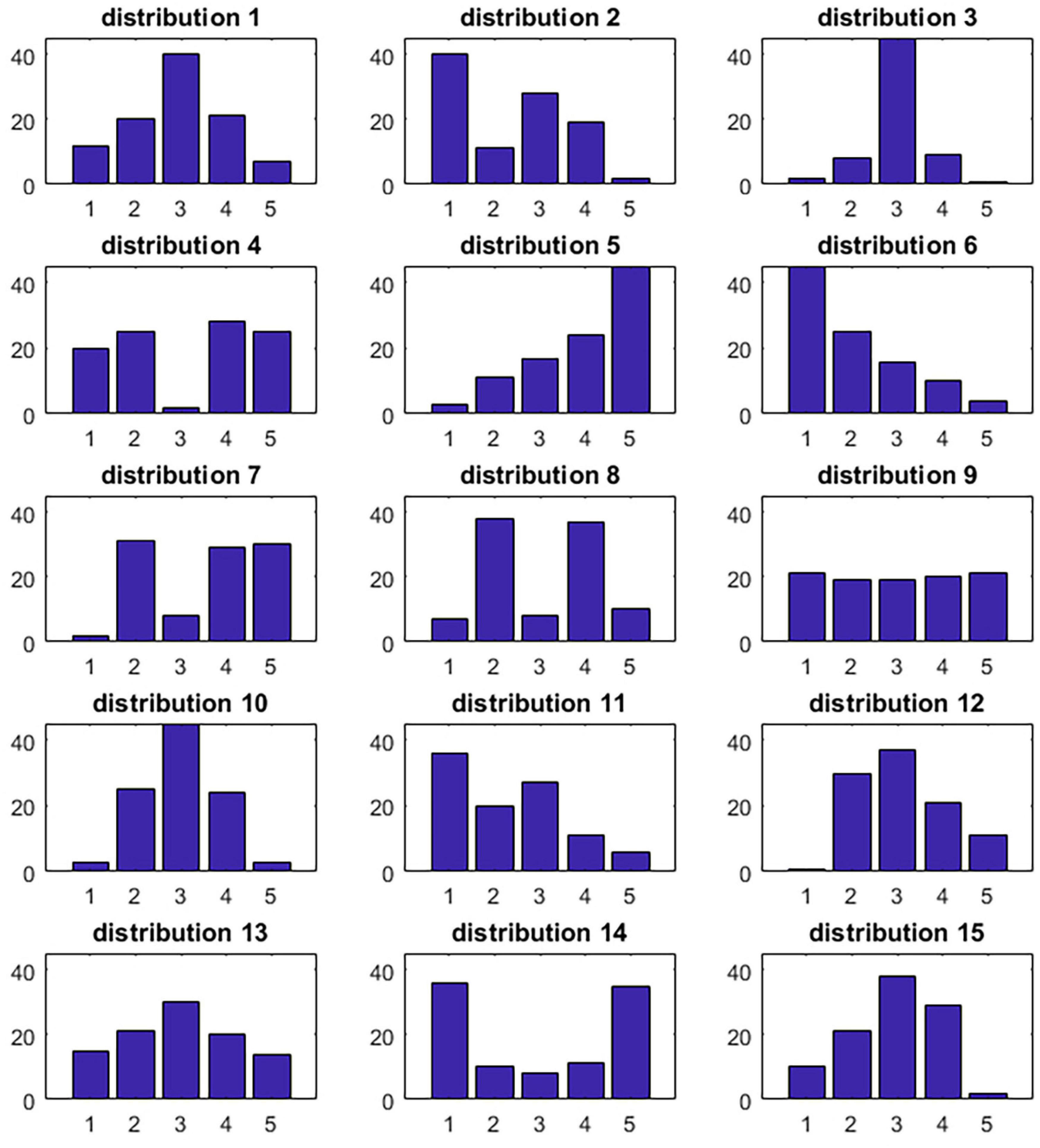}
    \caption{The 15 hypothetical distributions from \cite{koudenburg2021new} used for expert evaluation.}
    \label{fig:experts_histograms}
\end{figure}

\begin{figure}[H]
    \centering
    \includegraphics[width=0.85\linewidth]{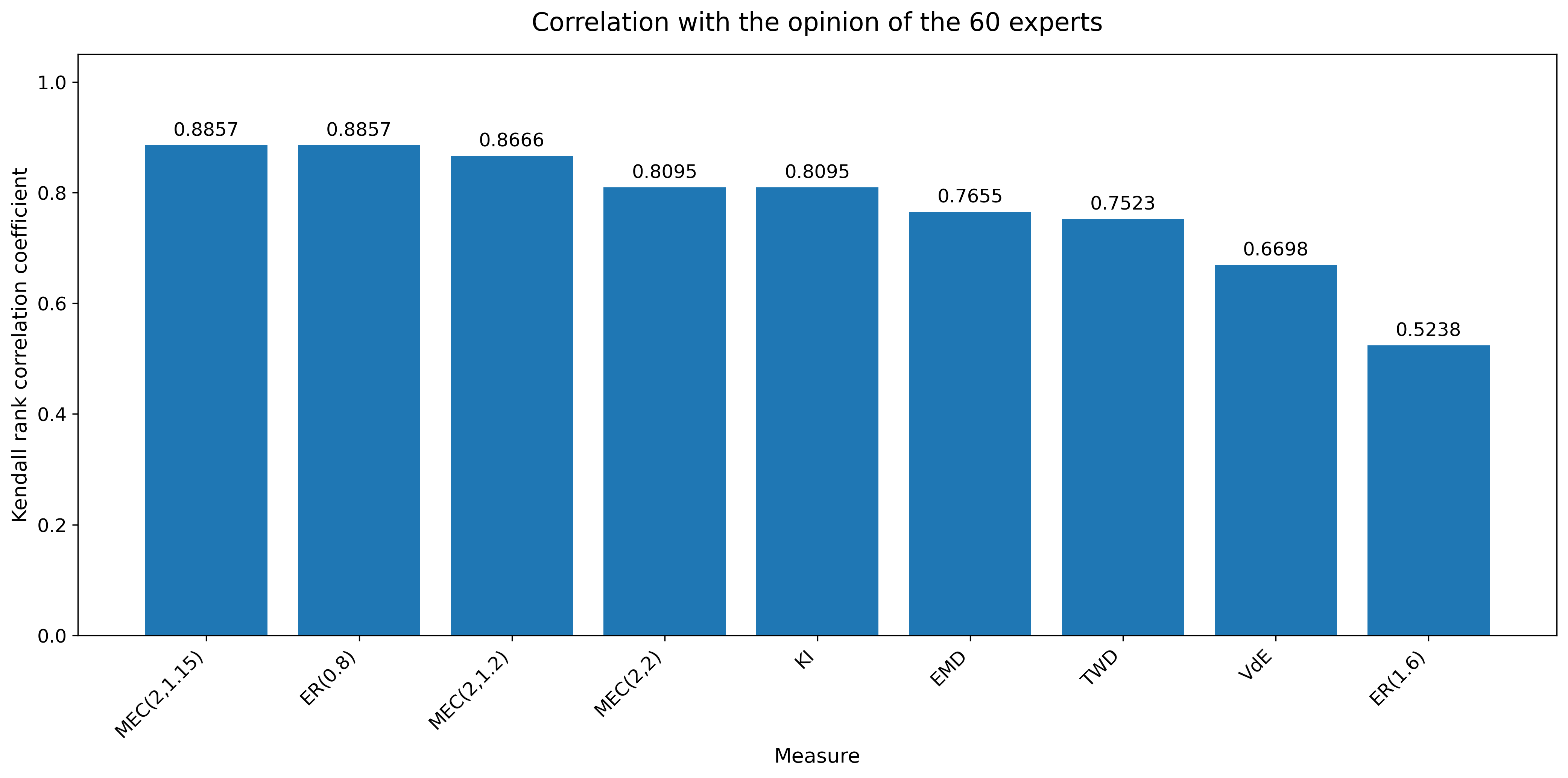}
    \caption{Correlation with expert judgments.}
    \label{fig:experts-opinion-correlation}
\end{figure}

The results highlight that $\MEC[2,1.15]$ and \ERLabel[0.8]{} show the highest concordance with expert evaluations, both achieving $\tau_{B}\approx 0.89$, followed by $\MEC[1,1]$ ($\tau_{B}=0.77$). In contrast, \VdELabel{} and \ERLabel[1.6]{} register lower correlations ($\tau_{B}=0.67$ and $\tau_{B}=0.52$, respectively). This divergence suggests that empirical agreement with expert judgments depends not only on broad normative motivation, but also on functional form, parameterization, and sensitivity to specific distributional patterns. This reinforces the need to complement axiomatic analysis with exhaustive empirical validation.

\section{Conclusions and Related Work}\label{sec:conclusion}

We have introduced MEC, a polarization measure defined as the minimum
effort required to bring an opinion distribution to consensus. The
measure rests on three connected mathematical foundations: a
transportation-theoretic interpretation as a Wasserstein-type distance
to the set of consensus distributions, a Minkowski-type characterization
as a weighted $L^\beta$-dispersion of the $\alpha$-power distribution,
and an axiomatic compatibility with the framework of Esteban and
Ray~\cite{esteban1994measurement}. The basic case $(\alpha,\beta)=(1,1)$
recovers the Earth Mover's Distance to the extremal distribution, and
the cases $(1,1)$ and $(1,2)$ specialize to mean absolute deviation
about the median and to a variance-like dispersion respectively. The
parameters $\alpha$ and $\beta$ admit a natural reading in terms of
group identification~\cite{TajfelTurner1979} and cognitive cost of
belief revision~\cite{Festinger1957}.

Beyond the scalar value, the optimization formulation produces an
optimal consensus point $y^*$, which existing pairwise and dispersion
measures do not provide. This additional structure has both
theoretical and practical use. On the theoretical side, $y^*$ anchors
the proofs of the Minority Principle and of the shifting-away
result, and supports a closed-form Tipping Point method that locates
the critical mass at which a partial transfer of opinions away from
$y^*$ ceases to reduce polarization and begins to increase it. The
analysis also reveals that polarization is not monotone in extremism,
nor in movement away from the optimal consensus point: whether such a
movement raises or lowers polarization depends on whether a whole group
or only a fragment of it is displaced. Empirically, MEC matches the
strongest parametrization of \ERLabel{} on the expert-judgment
benchmark of Koudenburg, Kiers, and Kashima~\cite{koudenburg2021new}, and
outperforms the dispersion-based alternatives \VdELabel{} and
\TWDLabel.

The proposal is part of a broader effort to bring tools from optimal
transport and convex analysis to the quantitative study of polarization,
a field that has been dominated by pairwise constructions in the spirit
of Gini and Esteban--Ray~\cite{esteban1994measurement,foster85}. The
Earth Mover's Distance has appeared in this literature as a measure of
distance between distributions~\cite{geher2020isometric,vallender1974calculation},
but, to our knowledge, has not previously been used to define
polarization itself as a distance to a structured target set. The
parametric framework $(\alpha,\beta)$ also relates to the
inequality-measurement tradition~\cite{foster85}, in which similar
exponents control sensitivity to mass and distance, and to the
sociological literature on group identification and cognitive
dissonance~\cite{TajfelTurner1979,Festinger1957}, which motivates the
specific role of the two parameters.

Several directions remain open. The current formulation treats opinions
as elements of $[0,1]$ and distributions as having finite support. A
natural extension is to general metric spaces, in which the ground
distance $|x-y|$ would be replaced by an arbitrary metric and the
domain by a suitable measurable space. The parametric family invites
further study of the dependence of $\MEC$ on $(\alpha,\beta)$, in
particular a theoretical justification of the empirically successful
parametrization $(2,1.15)$. The Tipping Point method, established here
for $\alpha=\beta=2$, may admit extensions to other parameter ranges.
The empirical validation rests on a synthetic benchmark; testing MEC
on real survey data, and on dynamic models of opinion update, would
clarify both its descriptive accuracy and its sensitivity to noise.
Finally, the optimization formulation suggests a connection to
mediation and policy framing through the optimal consensus point
$y^*$, which we have treated here only as a structural object but
which may admit substantive interpretation in applied settings.

\paragraph{Acknowledgments.}
The authors would like to thank the anonymous reviewers for their very
insightful comments on an earlier version of this paper.  This work
has been partially supported by the SGR project PROMUEVA (BPIN
2021000100160) under the supervision of Minciencias (Ministerio de
Ciencia Tecnolog\'ia e Innovaci\'on, Colombia),  by the CNRS project TOBIAS under the MITI interdisciplinary program, and the ECOSNORD Project BRAINS C26M02.

\bibliographystyle{plain}
\bibliography{main}

\newpage

\appendix
\begin{center}
{\fontsize{24}{48}\selectfont \bfseries Appendix}
\end{center}
\vspace{2em}

\section{Proofs for The MEC Measure}
\label{appendix:the-mec-measure-proofs}

\subsection{Earth Mover's Distance on Opinion Distributions}
\label{appendix:emd-definition}

We briefly recall the Earth Mover's Distance for opinion distributions and establish the identity used in the body to express $\ECEMD(y)$ as a transportation cost.

\begin{definition}[Transportation plan and \EMDLabel]\label{def:emd}
Let $\mu,\nu\in\Dset$ with $\Mass{\mu}=\Mass{\nu}$. A \emph{transportation plan} from $\mu$ to $\nu$ is a function $\gamma:[0,1]\times[0,1]\to\PositiveReals$ with finite support contained in $\support(\mu)\times\support(\nu)$, satisfying
\[
\sum_{y'\in\support(\nu)} \gamma(x,y')=\mu(x)
\qquad\text{for every } x\in\support(\mu),
\]
\[
\sum_{x'\in\support(\mu)} \gamma(x',y)=\nu(y)
\qquad\text{for every } y\in\support(\nu).
\]
The \emph{cost} of $\gamma$ is
\[
\mathrm{cost}(\gamma)\defsymbol \sum_{x,y} \gamma(x,y)\,|x-y|,
\]
and the \emph{Earth Mover's Distance} between $\mu$ and $\nu$ is
\[
\EMD(\mu,\nu)\defsymbol \min_{\gamma}\,\mathrm{cost}(\gamma),
\]
where the minimum is taken over all transportation plans from $\mu$ to $\nu$.
\end{definition}

\begin{lemma}[$\ECEMD$ as transportation cost]\label{lem:ecemd-as-emd}
For every $\pi\in\Dset$ and every $y\in[0,1]$,
\[
\ECEMD(y)=\EMD(\pi,\pi_y),
\]
where $\pi_y$ is the consensus distribution at $y$ with $\Mass{\pi_y}=\Mass{\pi}$.
\end{lemma}

\begin{proof}
Since $\pi_y$ is concentrated at the single point $y$, with $\pi_y(y)=\Mass{\pi}$, any transportation plan $\gamma$ from $\pi$ to $\pi_y$ has $\support(\gamma)\subseteq\support(\pi)\times\{y\}$. The marginal constraint at $y$ then reads
\[
\sum_{x\in\support(\pi)} \gamma(x,y)=\pi_y(y)=\Mass{\pi},
\]
and the marginal constraint at each $x\in\support(\pi)$ forces
\[
\gamma(x,y)=\pi(x).
\]
Hence there is a unique transportation plan from $\pi$ to $\pi_y$, given by $\gamma(x,y)=\pi(x)$ for $x\in\support(\pi)$ and zero elsewhere. Its cost is
\[
\sum_{x\in\support(\pi)} \pi(x)\,|x-y|=\ECEMD(y).
\]
Therefore $\EMD(\pi,\pi_y)=\ECEMD(y)$.
\end{proof}

\subsection{Proof of Theorem~\ref{th:mec-as-emd}}
\label{appendix:proof-mec-as-emd}

\mecAsEmd*
\begin{proof}
Let $m$ be any median of $\pi$. Notice that $m$ is also a median of $\pi^*$ because the set of medians of $\pi^*$ is $[0,1]$, so by Lemma~\ref{lem:ecemd-as-emd} we have $\MECEMD(\pi)=\EMD(\pi,\pi_m)$ and $\MECEMD(\pi^*)=\EMD(\pi^*,\pi_m)$.
Using these equations and Equation~\eqref{eq:mec-emd-pi-m}, Equation~\eqref{eq:mec-emd-1} can be rewritten as
\begin{equation}
\EMD(\pi^*, \pi_m)=\EMD(\pi^*, \pi) + \EMD(\pi,\pi_m),
\label{eq:mec-emd-2}
\end{equation}
which states that $\pi$ is an optimal transportation plan from $\pi^*$ to $\pi_m$ that moves masses from $\pi^*$ to the intermediate distribution $\pi$ and continues from $\pi$ to $\pi_m$.
To prove~\eqref{eq:mec-emd-2}, for any opinion distribution $\mu$, define the right-continuous quantile function as $Q_\mu(q) := \max\{x_0\in\support(\mu): \sum_{x<x_0} \mu(x) \leq q\}$.
On one hand, it is known~\cite{geher2020isometric,vallender1974calculation} that $\EMD$ can be expressed as
\begin{equation}
\EMD(\mu, \nu) = \int_0^{\Mass{\mu}} |Q_\mu(q)-Q_\nu(q)|\,dq, \quad \Mass{\mu}=\Mass{\nu}.
\label{eq:mec-emd-3}
\end{equation}
On the other hand, the quantile function $Q_\pi$ is sandwiched between $Q_{\pi^*}$ and $Q_{\pi_m}$ because for all $q\in[0,\Mass{\pi}/2]$, $Q_\pi(q)\in[0,m] = [Q_{\pi^*}, Q_{\pi_m}]$ and for all $q\in(\Mass{\pi}/2, \Mass{\pi}]$, $Q_\pi(q)\in[m,1] = [Q_{\pi_m}, Q_{\pi^*}]$.
Combining these two facts yields
\[
\begin{aligned}
  \EMD(\pi^*,\pi_m) &= \int_0^{\Mass{\pi^*}} |Q_{\pi^*}(q)-Q_{\pi_m}(q)|\,dq\\
  &= \int_0^{\Mass{\pi^*}} |Q_{\pi^*}(q)-Q_\pi(q)| + |Q_\pi(q)-Q_{\pi_m}(q)|\,dq\\
  &= \EMD(\pi^*,\pi) + \EMD(\pi,\pi_m).
\end{aligned}
\]
\end{proof}

\section{Proofs for Structural Properties of MEC}
\label{appendix:structural-properties-proofs}

\subsection{Convexity}
\subsubsection{Proof of Theorem~\ref{convexity}}
\label{appendix:proof-convexity}
\convexityThm*
\begin{proof}
From Def.~\ref{mec:def}, \(\beta \geq 1\). For any fixed \(x\in \support(\pi)\), consider 
\[
g_x(y)=|x-y|^\beta.
\]
It is a well-known result (see \cite{Rockafellar1970}) that \(g_x(y)\) is convex on \(\Reals\) for \(\beta\ge1\) and strictly convex for \(\beta>1\). Both convexity and strict convexity are preserved under multiplication by strictly positive constants. Hence, since for \(x\in \support(\pi)\) we have \(\pi(x)^\alpha>0\), the function 
\[
f_x(y)=\pi(x)^\alpha\, |x-y|^\beta
\]
is convex and, if \(\beta>1\), strictly convex as well. As \(\EC(y)\) is a finite sum of the functions \(f_x(y)\), and (strict) convexity is preserved under finite sums, it follows that \(\EC\) is convex for \(\beta\ge1\) and strictly convex for \(\beta>1\).
\end{proof}

\subsection{Characterization and Special Cases for MEC}
\subsubsection{Proof of Theorem~\ref{thm:mec-Lbeta}}
\label{appendix:proof-mec-lbeta}
\mecLbeta*
\begin{proof}
For $\alpha\ge 1$ and $\beta\ge 1$, the effort function is
\[
\EC(y)=\sum_{x\in\support(\pi)} \pi(x)^\alpha |x-y|^\beta
=\sum_{x\in\support(\pi)} \pi^\alpha(x)\,|x-y|^\beta.
\]
Factoring out $\Mass{\pi^\alpha}$ yields
\[
\EC(y)=\Mass{\pi^\alpha}\sum_{x\in\support(\pi)} \overline{\pi^\alpha}(x)\,|x-y|^\beta.
\]
Taking minima over $y\in[0,1]$ gives
\[
\MEC[\alpha,\beta](\pi)
=\min_{y\in[0,1]}\EC(y)
=\Mass{\pi^\alpha}\min_{y\in[0,1]}\sum_{x\in\support(\pi)} \overline{\pi^\alpha}(x)\,|x-y|^\beta
=\Mass{\pi^\alpha}\cdot D_\beta(\overline{\pi^\alpha}),
\]
as claimed.
\end{proof}

\subsubsection{Proof of Corollary~\ref{cor:mec-special-cases}}
\label{appendix:proof-cor-mec-special-cases}
\specialCases*
\begin{proof}
Let $\rho=\overline{\pi^\alpha}$. If $\beta=1$, the function $y\mapsto \sum_x \rho(x)\,|x-y|$ is minimized at any
median $m$ of $\rho$, hence
\[
\min_y \sum_x \rho(x)\,|x-y|=\sum_x \rho(x)\,|x-m|
=\mathrm{MAD}_{\mathrm{m}}(\rho).
\]
If $\beta=2$, the function $y\mapsto \sum_x \rho(x)\,(x-y)^2$ is minimized at
the mean $\mu$ of $\rho$, and the minimum equals
\[
\min_y \sum_x \rho(x)\,(x-y)^2=\sum_x \rho(x)\,(x-\mu)^2=\mathrm{Var}(\rho).
\]
Multiplying by $\Mass{\pi^\alpha}$ yields the stated identities.
\end{proof}

\subsection{Symmetry}

\subsubsection{Proof of Lemma~\ref{lem:mec-reflection}}
\label{appendix:proof-mec-reflection}

\mecReflection*
\begin{proof}
Let
\[
I_c \defsymbol [0,1]\cap[2c-1,2c].
\]
By definition of reflection,
\[
\support(\pi^{\mathrm{ref}(c)})
=
\{\,2c-x \mid x\in\support(\pi)\,\}
\]
and
\[
\pi^{\mathrm{ref}(c)}(2c-x)\defsymbol \pi(x)
\qquad\text{for all }x\in\support(\pi).
\]
Since $c$ is an admissible reflection center for $\pi$, for every
$x\in\support(\pi)$ we have $2c-x\in[0,1]$, hence
\[
2c-1\le x\le 2c.
\]
Therefore
\[
\support(\pi)\subseteq I_c
\qquad\text{and}\qquad
\support(\pi^{\mathrm{ref}(c)})\subseteq I_c.
\]
For every $y\in I_c$ we have $2c-y\in I_c\subseteq[0,1]$, and
\[
\begin{aligned}
\EC[\alpha,\beta][\pi^{\mathrm{ref}(c)}](y)
&=
\sum_{x\in\support(\pi^{\mathrm{ref}(c)})}
\pi^{\mathrm{ref}(c)}(x)^\alpha |x-y|^\beta \\
&=
\sum_{z\in\support(\pi)}
\pi(z)^\alpha |(2c-z)-y|^\beta \\
&=
\sum_{z\in\support(\pi)}
\pi(z)^\alpha |z-(2c-y)|^\beta \\
&=
\EC(2c-y).
\end{aligned}
\]
Since $\support(\pi)\subseteq I_c$, the minimum of $\EC$ over $[0,1]$
is attained in $I_c$, and the same holds for $\EC[\alpha,\beta][\pi^{\mathrm{ref}(c)}]$.
Because $y\mapsto 2c-y$ is a bijection from $I_c$ onto itself, we obtain
\[
\MEC[\alpha,\beta](\pi^{\mathrm{ref}(c)})
=
\MEC[\alpha,\beta](\pi).
\]
\end{proof}

\subsubsection{Proof of Corollary~\ref{cor:mec-shift}}
\label{appendix:proof-mec-shift}

\mecShift*
\begin{proof}
Observe that a shift can be expressed as the composition of two reflections.
If $c$ and $c+\tfrac{t}{2}$ are admissible reflection centers for the relevant
distributions, then
\[
\pi^{\mathrm{shift}(t)}
=
\left(\pi^{\mathrm{ref}(c)}\right)^{\mathrm{ref}(c+t/2)},
\]
since for every $x\in\support(\pi)$,
\[
2\!\left(c+\frac{t}{2}\right)-(2c-x)=x+t.
\]
In particular, if
\[
c=\frac{x_1+x_n}{2},
\]
where $x_1=\min\support(\pi)$ and $x_n=\max\support(\pi)$, then $c$ is an
admissible reflection center for $\pi$; and if $t$ is an admissible shift,
then $c+\tfrac{t}{2}$ is an admissible reflection center for
$\pi^{\mathrm{ref}(c)}$.
\end{proof}

\subsubsection{Proof of Lemma~\ref{lem:symmetry-mec}}
\label{appendix:proof-symmetry-mec}

\symmetryMec*
\begin{proof}
Recall that
\[
\EC(y)=\sum_{x\in\support(\pi)} \pi(x)^\alpha |x-y|^\beta.
\]
Since $\pi$ is symmetric about $c$, for every $x\in\support(\pi)$ we have
$2c-x\in\support(\pi)$ and
\[
\pi(x)=\pi(2c-x).
\]
Hence every point of $\support(\pi)$ distinct from $c$ belongs to a symmetric
pair $\{x,2c-x\}$, and if $c\in\support(\pi)$ it forms a singleton.

We first show that $\EC$ is symmetric about $c$. Let $t$ be such that
$c\pm t\in[0,1]$. For each symmetric pair $\{x,2c-x\}$, using
$\pi(x)=\pi(2c-x)$, we have
\[
\begin{aligned}
&\pi(x)^\alpha |x-(c+t)|^\beta
+\pi(2c-x)^\alpha |(2c-x)-(c+t)|^\beta \\
&\qquad
=\pi(x)^\alpha |x-(c+t)|^\beta
+\pi(x)^\alpha |x-(c-t)|^\beta .
\end{aligned}
\]
Interchanging $t$ and $-t$ leaves this quantity unchanged, so the
contribution of each symmetric pair to $\EC(c+t)$ equals its contribution
to $\EC(c-t)$. If $c\in\support(\pi)$, then the singleton contributes
\[
\pi(c)^\alpha |c-(c+t)|^\beta
=
\pi(c)^\alpha |t|^\beta
=
\pi(c)^\alpha |c-(c-t)|^\beta,
\]
which is also the same at $c+t$ and $c-t$. Summing over all support points,
we obtain
\[
\EC(c+t)=\EC(c-t).
\]
Therefore $\EC$ is symmetric about $c$.

If $\beta>1$, then by Theorem~\ref{convexity}, $\EC$ is strictly convex on
$[0,1]$, so it has a unique optimal consensus point. Since $\EC$ is symmetric about $c$,
that unique optimal consensus point must be $c$. Hence
\[
\MEC[\alpha,\beta](\pi)=\EC(c).
\]

Assume now that $\beta=1$. By Corollary~\ref{cor:mec-special-cases}, the
optimal consensus points of $\pi$ are precisely the medians of $\overline{\pi^\alpha}$. Since
\[
\pi(x)^\alpha=\pi(2c-x)^\alpha
\qquad\text{for all }x\in[0,1],
\]
the distribution $\pi^\alpha$ is also symmetric about $c$. Therefore the
total $\pi^\alpha$-mass strictly to the left of $c$ equals the total
$\pi^\alpha$-mass strictly to the right of $c$, and thus both are at most
$\Mass{\pi^\alpha}/2$. Hence $c$ is a median of $\pi^\alpha$, so $\EC(c)$
attains the minimum. Therefore again
\[
\MEC[\alpha,\beta](\pi)=\EC(c).
\]

Thus, in all cases,
\[
\MEC[\alpha,\beta](\pi)
=
\EC(c)
=
\sum_{x\in\support(\pi)} \pi(x)^\alpha |x-c|^\beta.
\]
\end{proof}

\subsubsection{Proof of Theorem~\ref{th:mec-pi-star}}
\label{appendix:proof-mec-pi-star}

\mecPiStar*
\begin{proof}
Since $\pi^*(0)=\pi^*(1)=\Mass{\pi^*}/2$, the distribution $\pi^*$ is symmetric about $1/2$. Hence, by Lemma~\ref{lem:symmetry-mec},
\[
\MEC(\pi^*)=\EC[\alpha,\beta][\pi^*](1/2).
\]
Evaluating the effort at $1/2$ gives
\[
\EC[\alpha,\beta][\pi^*](1/2)
=
\left(\frac{\Mass{\pi^*}}{2}\right)^\alpha\!\left(\frac12\right)^\beta
+
\left(\frac{\Mass{\pi^*}}{2}\right)^\alpha\!\left(\frac12\right)^\beta
=
2\left(\frac{\Mass{\pi^*}}{2}\right)^\alpha\left(\frac12\right)^\beta.
\]
\end{proof}

\subsection{Population Scaling}

\subsubsection{Proof of Theorem~\ref{th:multbyscalar}}
\label{appendix:proof-population-scaling}

\multByScalar*
\begin{proof}
By definition,
\[
\MEC(\lambda\pi)
=
\min_{y\in[0,1]}
\sum_{i=1}^{n}(\lambda\pi_i)^{\alpha}|x_i-y|^{\beta}.
\]
Since
\[
(\lambda\pi_i)^{\alpha}=\lambda^{\alpha}\pi_i^{\alpha},
\]
it follows that
\[
\MEC(\lambda\pi)
=
\min_{y\in[0,1]}
\lambda^{\alpha}\sum_{i=1}^{n}\pi_i^{\alpha}|x_i-y|^{\beta}.
\]
Because $\lambda^{\alpha}$ does not depend on $y$, it can be factored out of the minimum:
\[
\MEC(\lambda\pi)
=
\lambda^{\alpha}
\min_{y\in[0,1]}
\sum_{i=1}^{n}\pi_i^{\alpha}|x_i-y|^{\beta}.
\]
Therefore,
\[
\MEC(\lambda\pi)=\lambda^{\alpha}\MEC(\pi).
\]
\end{proof}

\subsubsection{Proof of Corollary~\ref{cor:condition-H-normalization}}
\label{appendix:proof-condition-H-normalization}

\conditionH*
\begin{proof}
(1) By definition of normalization,
\[
\overline{\pi}=\frac{1}{\Mass{\pi}}\,\pi.
\]
Applying Th.~\ref{th:multbyscalar} with
\[
\lambda=\frac{1}{\Mass{\pi}},
\]
we obtain
\[
\MEC(\overline{\pi})
=
\MEC\!\left(\frac{1}{\Mass{\pi}}\pi\right)
=
\left(\frac{1}{\Mass{\pi}}\right)^\alpha \MEC(\pi).
\]
Therefore,
\[
\MEC(\overline{\pi})
=
\frac{\MEC(\pi)}{(\Mass{\pi})^\alpha},
\]
as claimed.

(2) From Th.~\ref{th:multbyscalar} we have
\[
\MEC(\lambda\pi)=\lambda^{\alpha}\MEC(\pi)
\qquad\text{and}\qquad
\MEC(\lambda\pi')=\lambda^{\alpha}\MEC(\pi').
\]
Since $\lambda^{\alpha}>0$, multiplying the inequality
\[
\MEC(\pi)>\MEC(\pi')
\]
by $\lambda^{\alpha}$ preserves the order, yielding
\[
\MEC(\lambda\pi)>\MEC(\lambda\pi').
\]
Hence MEC satisfies Condition~H.
\end{proof}

\section{Proofs for Theoretical Applications}
\label{appendix:theoretical-applications-proofs}

\subsection{A Minority Principle}

\subsubsection{Proof of Lemma~\ref{lem:yopt-closest-to-largest-general}}
\label{appendix:proof-yopt-closest-to-largest-general}

\yoptClosest*
\begin{proof}
Since \(\beta>1\), the function \(\EC[\alpha,\beta][\pi]\) is strictly convex, hence it has a unique optimal consensus point \(y^*\). On \([l,r]\), we have \(\EC[\alpha,\beta][\pi](y)=A^\alpha (y-l)^\beta + B^\alpha (r-y)^\beta\), and for \(y\in(l,r)\),
\[
\frac{d}{dy}\EC[\alpha,\beta][\pi](y)
=
\beta A^\alpha (y-l)^{\beta-1}
-
\beta B^\alpha (r-y)^{\beta-1}.
\]
Moreover, \(\frac{d}{dy}\EC[\alpha,\beta][\pi](l^+)=-\beta B^\alpha (r-l)^{\beta-1}<0\) and \(\frac{d}{dy}\EC[\alpha,\beta][\pi](r^-)=\beta A^\alpha (r-l)^{\beta-1}>0\). Hence the unique optimal consensus point \(y^*\) lies in \((l,r)\). Since \(\EC[\alpha,\beta][\pi]\) is strictly convex, its derivative is strictly increasing on \((l,r)\).

Suppose first that \(y^*>\frac{l+r}{2}\). Since \(\frac{d}{dy}\EC[\alpha,\beta][\pi](y^*)=0\), strict monotonicity of the derivative gives \(0>\frac{d}{dy}\EC[\alpha,\beta][\pi](\frac{l+r}{2})\). Thus \(0>\beta\left(\frac{r-l}{2}\right)^{\beta-1}(A^\alpha-B^\alpha)\), and therefore \(A^\alpha-B^\alpha<0\). Hence \(B^\alpha>A^\alpha\), so \(B>A\).

If \(y^*<\frac{l+r}{2}\), the same argument yields \(0<\frac{d}{dy}\EC[\alpha,\beta][\pi](\frac{l+r}{2})\), hence \(A^\alpha-B^\alpha>0\), and therefore \(B<A\).

Finally, if \(y^*=\frac{l+r}{2}\), then \(0=\frac{d}{dy}\EC[\alpha,\beta][\pi](\frac{l+r}{2})=\beta\left(\frac{r-l}{2}\right)^{\beta-1}(A^\alpha-B^\alpha)\), so \(A^\alpha=B^\alpha\), and therefore \(A=B\).
\end{proof}

\subsubsection{Proof of Lemma~\ref{lem:greatest-contribution-from-smallest-general}}
\label{appendix:proof-greatest-contribution-from-smallest-general}

\greatestContribution*
\begin{proof}
We prove (1); the other two cases are analogous. Suppose \(y^*>\frac{l+r}{2}\). Then \(y^*\in(l,r)\), so \(\EC[\alpha,\beta][\pi]\) is differentiable at \(y^*\). Since \(y^*\) is an optimal consensus point, we have
\[
\frac{d}{dy}\EC[\alpha,\beta][\pi](y^*)=0.
\]
Thus
\[
A^\alpha (y^*-l)^{\beta-1}=B^\alpha (r-y^*)^{\beta-1}.
\]
Moreover, \(y^*>\frac{l+r}{2}\) implies \(y^*-l>r-y^*\). Multiplying the equal sides above by these unequal positive factors, we obtain
\[
A^\alpha (y^*-l)^\beta > B^\alpha (r-y^*)^\beta.
\]
If \(y^*<\frac{l+r}{2}\), then \(y^*-l<r-y^*\), and the same argument yields
\[
A^\alpha (y^*-l)^\beta < B^\alpha (r-y^*)^\beta.
\]
Finally, if \(y^*=\frac{l+r}{2}\), then \(y^*-l=r-y^*\), and therefore
\[
A^\alpha (y^*-l)^\beta = B^\alpha (r-y^*)^\beta.
\]
\end{proof}

\subsubsection{Proof of Theorem~\ref{thm:minority-principle-general}}
\label{appendix:proof-minority-principle-general}

\minorityPrinciple*
\begin{proof}
Assume first that \(\beta>1\). Suppose \(A<B\). Then \(B>A\). By Lem.~\ref{lem:yopt-closest-to-largest-general}, we cannot have \(y^*<\frac{l+r}{2}\), for that would imply \(B<A\), nor can we have \(y^*=\frac{l+r}{2}\), for that would imply \(B=A\). Hence \(y^*>\frac{l+r}{2}\). It now follows from Lem.~\ref{lem:greatest-contribution-from-smallest-general} that
\[
A^\alpha (y^*-l)^\beta > B^\alpha (r-y^*)^\beta.
\]
The case \(A>B\) is analogous. Indeed, if \(A>B\), then \(B<A\). By Lem.~\ref{lem:yopt-closest-to-largest-general}, we cannot have \(y^*>\frac{l+r}{2}\), nor can we have \(y^*=\frac{l+r}{2}\). Hence \(y^*<\frac{l+r}{2}\), and therefore, by Lem.~\ref{lem:greatest-contribution-from-smallest-general},
\[
A^\alpha (y^*-l)^\beta < B^\alpha (r-y^*)^\beta.
\]
Assume now that \(\beta=1\). By Cor.~\ref{cor:mec-special-cases}, \(y^*\) is a median of \(\pi^\alpha\). Since \(A\neq B\), we also have \(A^\alpha\neq B^\alpha\), so this median is unique and coincides with the location of the larger mass. Thus, if \(A<B\), then \(A^\alpha<B^\alpha\), hence \(y^*=r\). Therefore \(A^\alpha (y^*-l)=A^\alpha(r-l)>0\), while \(B^\alpha(r-y^*)=0\). Hence
\[
A^\alpha (y^*-l)^\beta > B^\alpha (r-y^*)^\beta.
\]
The case \(A>B\) is symmetric.
\end{proof}

\section{Proof that MEC is a Polarization Measure}
\label{Appendix:mec-polarization-measure}
\def\sign{{\textrm{sign}}}
\def\half{{\frac{1}{2}}}

This appendix proves Theorem~\ref{th:mec-is-a-polarization-measure}. The structure of the proof was outlined in Section~\ref{sec:shifting-away}. The key technical results are: (i) Theorem~\ref{thm:mec-increases:body}, which establishes that any non-median shift away from consensus strictly increases polarization; (ii) Theorem~\ref{thm:two_massed_distribution:body}, which shows that among two-point distributions supported at the extremes the balanced one is maximal; and (iii) Theorem~\ref{thm:mec-maximal}, which assembles the previous results into the full claim.

We use the notation $\pi_{a\to b}$ from Definition~\ref{def:movement} throughout.

\subsection{Shifts Away from Consensus Strictly Increase MEC}

\begin{theorem}\label{thm:mec-increases-a1-b1}
Let $\alpha = 1$ and $\beta=1$. For any distribution $\pi$ and $a\in\support(\pi)$, if $\pi_{a\to b}$ is a shift away from consensus and is not a median-shift (Def.~\ref{def:movement}), then $\MEC(\pi_{a\to b}) > \MEC(\pi)$.
\end{theorem}
\begin{proof}
Let $[l,r]$ be the optimal consensus interval for $\pi$ for $\alpha=\beta=1$. Consider the case $r\le a<b$, as the mirror case $b<a\le l$ follows by symmetry.
We will first show that $r$ is a median for $\pi_{a\to b}$.

If $l<r$, then, given that $r$ is a median for $\pi$ and $a\ge r$, we have $\sum_{x\ge r}\pi_{a\to b}(x)=\sum_{x\ge r}\pi(x)=\frac{\Mass{\pi}}{2}=\frac{\Mass{\pi_{a\to b}}}{2}$. So $r$ is a median for $\pi_{a\to b}$.

If $l=r$, then $\sum_{x\ge r} \pi(x)> \Mass{\pi}/2 > \sum_{x > r} \pi(x)$. Since $\pi_{a\to b}$ is not a median-shift, $a>r$, and therefore
\[
\sum_{x\ge r}\pi_{a\to b}(x)=\sum_{x\ge r} \pi(x),
\qquad
\sum_{x> r}\pi_{a\to b}(x)=\sum_{x> r} \pi(x).
\]
This implies $\sum_{x\ge r} \pi_{a\to b}(x)> \Mass{\pi_{a\to b}}/2 > \sum_{x > r} \pi_{a\to b}(x)$, so $r$ is a median for $\pi_{a\to b}$.

Since $r$ is a median for both $\pi$ and $\pi_{a\to b}$, then
\[
\begin{aligned}
\MEC(\pi_{a\to b}) - \MEC(\pi) &= \EC[\alpha,\beta][\pi_{a\to b}](r) - \EC[\alpha,\beta][\pi](r)\\
&=\pi_{a\to b}(b) |b-r| - \pi(a) |a-r| - \pi(b) |b-r|\\
&=\pi(a) (b-a) > 0.
\end{aligned}
\]
\end{proof}

\begin{theorem}\label{thm:mec-increases-a1-b}
Let $\alpha = 1$ and $\beta>1$. Let $\pi\in\Dset$ and $a\in\support(\pi)$. If $\pi_{a\to b}$ is a shift away from consensus (Def.~\ref{def:movement}), then $\MEC(\pi_{a\to b})>\MEC(\pi)$.
\end{theorem}
\begin{proof}
Let $y^*$ be the optimal consensus point for $\pi$.
Consider the case $y^*\le a<b$, as the mirror case $b<a\le y^*$ follows by symmetry.
To prove that $\MEC$ increases as $b$ increases away from $y^*$, it suffices to show that $\frac{\partial}{\partial b} \MEC(\pi_{a\to b})>0$.

The Envelope Theorem states that
\[
\frac{\partial}{\partial b} \MEC(\pi_{a\to b})
=\left.\frac{\partial}{\partial b} \EC[\alpha,\beta][\pi_{a\to b}](y)\right|_{y=y^*_{a\to b}}\hspace{-2.5em} = \pi(a)\, \beta |b-y^*_{a\to b}|^{\beta-1} \sign(b-y^*_{a\to b}),
\]
so it remains to show that $b > y^*_{a\to b}$.

From the optimality of $y^*$ and $y^*_{a\to b}$, and letting $f(t):=\beta |t|^{\beta-1} \sign(t)$ and  $g(y) := \sum_{x \ne a} \pi(x) f(y-x)$, it holds that
\begin{align}
0 &= \left.\frac{\partial}{\partial y} \EC[\alpha,\beta][\pi](y)\right|_{y=y^*}
= g(y^*) + \pi(a) f(y^*-a),\label{eq:894}\\
0 &= \left.\frac{\partial}{\partial y} \EC[\alpha,\beta][\pi_{a\to b}](y)\right|_{y=y^*_{a\to b}}
= g(y^*_{a\to b}) + \pi(a) f(y^*_{a\to b}-b).\label{eq:895}
\end{align}

Subtracting Equations~\eqref{eq:894} and~\eqref{eq:895} produces
\begin{equation}
\pi(a) (f(y^*-a) - f(y^*_{a\to b}-b)) = g(y^*_{a\to b}) - g(y^*).
\label{eq:896}
\end{equation}
Note that $f$ is strictly increasing, hence so is $y\mapsto f(y-x)$ for every $x$, and so is the function $g$.
Using the monotonicity of $f$ and $g$ on Equation~\eqref{eq:896} yields the following sign equalities
\def\eqsign{\stackrel{\sign}{=}}
\begin{align}
    f(y^*-a) - f(y^*_{a\to b}-b) &\eqsign g(y^*_{a\to b}) - g(y^*)\nonumber\\
    (y^*-a) - (y^*_{a\to b}-b) &\eqsign y^*_{a\to b} - y^*\nonumber\\
    (b-a) - (y^*_{a\to b}-y^*) &\eqsign y^*_{a\to b} - y^*\label{eq:897}
\end{align}

Equation~\eqref{eq:897} has the form $z-w \eqsign w$ with $z>0$, which reduces to $0<w<z$. Therefore, $0 < y^*_{a\to b}-y^* < b-a$, and $b>y^*_{a\to b}>y^*$.
\end{proof}

\begin{theorem}\label{thm:mec-increases-aux}
For any $\alpha,\beta\geq 1$, any distribution $\pi$, any $a\in\support(\pi)$ and any $b$, it holds that
\[
\MEC[1,\beta]((\pi^\alpha)_{a\to b}) \leq \MEC[1,\beta]((\pi_{a\to b})^\alpha).
\]
\end{theorem}
\begin{proof}
For any $c\ge 0$, let $\mu_c$ denote the distribution given by
\[
\mu_c(x)=\pi^\alpha(x) \text{ for all } x\notin \set{a,b},
\qquad
\mu_c(a)=0,
\qquad
\mu_c(b)=c.
\]
Then for $c_0:=(\pi^\alpha)_{a\to b}(b)$ we have $\mu_{c_0}=(\pi^\alpha)_{a\to b}$, and for $c_1:= (\pi_{a\to b})^\alpha(b)$ we have $\mu_{c_1} = (\pi_{a\to b})^\alpha$.

Since
\[
 c_0=\pi(a)^\alpha + \pi(b)^\alpha \leq (\pi(a)+\pi(b))^\alpha = c_1,
\]
it suffices to show that the mapping $g(c)=\MEC[1,\beta](\mu_c)$ is non-decreasing over the interval $c\in[c_0, c_1]$. Applying the Envelope Theorem to $g$ and $f(c,y) := \EC[1,\beta][\mu_c](y)$ yields $g'(c) =\frac{\partial}{\partial c} f(c, y)|_{y=y^*_c}$, where $y^*_c$ is an optimal consensus point of $\mu_c$.
So $g'(c) = \frac{\partial}{\partial c} c |b-y^*_c|^\beta = |b-y^*_c|^\beta \geq 0$.
\end{proof}

\shiftingAway*
\begin{proof}
Note that
\[
\MEC(\pi) = \MEC[1,\beta](\pi^\alpha)
\qquad\text{and}\qquad
\MEC(\pi_{a\to b})=\MEC[1,\beta]((\pi_{a\to b})^\alpha).
\]
Hence it suffices to show that 
\[
\MEC[1,\beta](\pi^\alpha) < \MEC[1,\beta]((\pi^\alpha)_{a\to b}) \leq \MEC[1,\beta]((\pi_{a\to b})^\alpha).
\]
The first inequality was shown in Theorems~\ref{thm:mec-increases-a1-b1} and~\ref{thm:mec-increases-a1-b} (for $\beta=1$ and $\beta>1$ respectively), and the second in Theorem~\ref{thm:mec-increases-aux}.
\end{proof}

\subsection{Two-Point Polarization}

\begin{lemma}\label{lem:two-point-max-simple}
Let \(k>0\) and \(A\in[0,k]\). Let \(\rho_A\in\Dset_k\) be the two-point distribution at $\{0,1\}$ with $\rho_A(0)=A$ and $\rho_A(1)=k-A$, and write \(g(A)\defsymbol \MEC(\rho_A)\). Then \(g(A)\le g(k/2)\) for every \(A\in[0,k]\).
\end{lemma}

\begin{proof}
By definition,
\[
g(A)=\min_{y\in[0,1]}\,\bigl(A^\alpha y^\beta +(k-A)^\alpha(1-y)^\beta\bigr).
\]
Write $f_A(y)\defsymbol A^\alpha y^\beta +(k-A)^\alpha(1-y)^\beta=\EC[\alpha,\beta][\rho_A](y)$. If \(A=0\) or \(A=k\), then $\rho_A$ is a consensus distribution, so \(g(A)=0\) and the conclusion is immediate. Thus assume \(0<A<k\).

If \(\beta=1\), then
\[
g(A)=\min\{A^\alpha,(k-A)^\alpha\},
\]
which is plainly maximal at \(A=k/2\).

Assume now that \(\beta>1\). By Theorem~\ref{convexity}, $f_A$ is strictly convex, so $\rho_A$ has a unique optimal consensus point \(y_A\in(0,1)\) satisfying
\[
A^\alpha y_A^{\beta-1}=(k-A)^\alpha(1-y_A)^{\beta-1}.
\]
Using this identity, we obtain
\[
\begin{aligned}
g(A)
&=A^\alpha y_A^\beta +(k-A)^\alpha(1-y_A)^\beta \\
&=A^\alpha y_A^\beta + A^\alpha y_A^{\beta-1}(1-y_A) \\
&=A^\alpha y_A^{\beta-1}.
\end{aligned}
\]
Similarly,
\[
g(A)=(k-A)^\alpha(1-y_A)^{\beta-1}.
\]
Multiplying these two expressions gives
\[
g(A)^2=[A(k-A)]^\alpha [y_A(1-y_A)]^{\beta-1}.
\]
Now \(A(k-A)\le (k/2)^2\) and \(y_A(1-y_A)\le 1/4\). Hence
\[
g(A)^2
\le
\left(\frac{k}{2}\right)^{2\alpha}\left(\frac14\right)^{\beta-1}.
\]
On the other hand, when \(A=k/2\) the distribution $\rho_{k/2}$ is symmetric about $1/2$, so by Lemma~\ref{lem:symmetry-mec} its optimal consensus point is $y=1/2$, and therefore
\[
g(k/2)
=
2\left(\frac{k}{2}\right)^\alpha \left(\frac12\right)^\beta
=
\left(\frac{k}{2}\right)^\alpha \left(\frac12\right)^{\beta-1}.
\]
Thus
\[
g(A)^2\le g(k/2)^2.
\]
Since \(g(A)\ge 0\), it follows that
\[
g(A)\le g(k/2).
\]
\end{proof}

\twoPointPolarization*
\begin{proof}
By Corollary~\ref{cor:condition-H-normalization}, we may assume $\Mass{\pi_A}=1$, so $\pi_A(0)=A\in[0,1]$ and $\pi_A(1)=1-A$. In the notation of Lemma~\ref{lem:two-point-max-simple} with $k=1$, we have $\pi_A=\rho_A$, hence $\MEC(\pi_A)=g(A)$. The lemma gives $g(A)\le g(1/2)$ for all $A\in[0,1]$, with equality at $A=1/2$. Since $\pi_A^*$ corresponds precisely to $A=1/2$, we conclude that $\MEC(\pi_A)\le \MEC(\pi_A^*)$.
\end{proof}

\subsection{Proof of Theorem~\ref{th:mec-is-a-polarization-measure}}

We first establish the auxiliary result that $\pi^*$ maximizes $\MEC$.

\begin{theorem}\label{thm:mec-maximal}
For any $\alpha,\beta\geq 1$ and any $\pi\in\Dset$, we have $\MEC(\pi)\le\MEC(\pi^*)$.
\end{theorem}
\begin{proof}
If $\pi$ is a consensus distribution, then $\MEC(\pi^*)\geq 0 = \MEC(\pi)$, so we may assume $|\support(\pi)|\geq 2$.
It suffices to find a distribution $\mu$ with $\support(\mu)=\set{0,1}$ and $\Mass{\mu}=\Mass{\pi}$ and $\MEC(\mu)\geq \MEC(\pi)$, since applying Theorem~\ref{thm:two_massed_distribution:body} to $\mu$ yields
\[
\MEC(\pi)\leq \MEC(\mu) \leq \MEC(\pi^*).
\]

If $\beta>1$, applying Theorem~\ref{thm:mec-increases:body} successively to $\pi$ with any $a\in\support(\pi)\setminus\{0,1\}$ and $b\in\set{0,1}$ chosen appropriately to move away from consensus yields, after finitely many steps, a distribution $\mu$ with $\support(\mu)=\{0,1\}$ and $\MEC(\mu)\ge\MEC(\pi)$. The conclusion follows.

If $\beta=1$, the same process leads to a distribution $\mu$ with $\support(\mu)\subseteq\set{0,y^*,1}$ for some $y^*\in(0,1)$, where the case $\support(\mu)=\set{0,1}$ corresponds to no stalling. If $\support(\mu)=\set{0,1}$, the proof concludes with Theorem~\ref{thm:two_massed_distribution:body} as above.

Otherwise, $y^*\in\support(\mu)$ and the process stalled because $y^*$ is the unique optimal consensus point of $\mu$. By Cor.~\ref{cor:mec-special-cases}, $y^*$ is then the unique median of $\mu^\alpha$. Therefore, $\mu(0),\mu(1)<\Mass{\pi}/2=\pi^*(0)=\pi^*(1)$, since otherwise $0$ or $1$ would also be a median of $\mu^\alpha$, contradicting uniqueness.

Since $y^*$ is also a median for $\pi^*$ (every point in $[0,1]$ is), then 
\[
\begin{aligned}
\MEC[\alpha,1](\pi)&\leq \MEC[\alpha,1](\mu)
 = \EC[\alpha,1][\mu](y^*)\\
&= \mu(0)^\alpha (y^*-0) + \mu(1)^\alpha (1-y^*)\\
&\leq \pi^*(0)^\alpha (y^*-0) + \pi^*(1)^\alpha (1-y^*)\\
&=\EC[\alpha,1][\pi^*](y^*)=\MEC(\pi^*),
\end{aligned}
\]
where $\mu(0)^\alpha (y^*-0)$ is interpreted as $0$ when $0\notin\support(\mu)$, and similarly for the term at $1$.
\end{proof}

\mecIsPolarization*
\begin{proof}
The first condition of Def.~\ref{pol:def} holds because $\EC(y)=0$ at the consensus point of any consensus distribution, hence $\MEC(\pi)=0$ when $\pi$ is a consensus distribution. The second condition is exactly Theorem~\ref{thm:mec-maximal}.
\end{proof}

\section{Proofs for \texorpdfstring{\ERLabel}{ER} Axioms and Property 1}
\label{appendix:MEC_proofs}

\subsection{Proof of Theorem~\ref{th:mec-axiom-1} (\texorpdfstring{\ERLabel}{ER} Axiom~1)}
\label{appendix:proof-mec-axiom-1}

\mecAxiomOne*
\begin{proof}
\proofpart{Setup}
Fix $p>0$, $x>0$, write $h\defsymbol y-x$, and set
\[
  \pi=((p,q,q),(0,x,x+h)),
  \qquad
  \tilde\pi=((p,2q),(0,x+h/2)).
\]
Define
\[
  E(z)\defsymbol \EC[\alpha,\beta][\pi](z)
  =p^\alpha z^\beta+q^\alpha\lvert x-z\rvert^\beta+q^\alpha\lvert x+h-z\rvert^\beta,
\]
\[
  \widetilde E(z)\defsymbol \EC[\alpha,\beta][\tilde\pi](z)
  =p^\alpha z^\beta+(2q)^\alpha\left\lvert x+\tfrac{h}{2}-z\right\rvert^\beta,
\]
so that $\MEC(\pi)=\min_{z\in[0,1]}E(z)$ and $\MEC(\tilde\pi)=\min_{z\in[0,1]}\widetilde E(z)$.

\proofpart{Central identity and decomposition}

For $z\in(0,x)$, set $A\defsymbol x-z>0$, so that $\lvert x-z\rvert=A$ and
$\lvert x+h-z\rvert=A+h$. Then
\begin{equation}
  \label{eq:central}
  E(z)-\widetilde{E}(z)=q^\alpha\,\Delta(A,h),
\end{equation}
where we define
\[
  \Delta(A,h)\defsymbol A^\beta+(A+h)^\beta-2^\alpha\!\left(A+\tfrac{h}{2}\right)^\beta.
\]
The key to the proof is the decomposition
\begin{equation}
  \label{eq:decomp}
  \Delta(A,h)
  =\underbrace{\left[A^\beta+(A+h)^\beta-2\!\left(A+\tfrac{h}{2}\right)^\beta\right]}_{J(A,h)\;\geq\;0}
  -(2^\alpha-2)\!\left(A+\tfrac{h}{2}\right)^\beta.
\end{equation}
The term $J(A,h)\geq 0$ is the Jensen remainder for the convex function
$t\mapsto t^\beta$: it vanishes identically when $\beta=1$, and is strictly
positive when $\beta>1$ and $h>0$. The subtracted term is strictly positive
when $\alpha>1$, and vanishes when $\alpha=1$.

The sign pattern of $\Delta$ in each regime is therefore:

\begin{equation}
  \label{eq:signs}
  \begin{aligned}
    \alpha>1,\;\beta=1\colon\quad &\Delta=-(2^\alpha-2)\!\left(A+\tfrac h2\right)<0,\\
    \alpha>1,\;\beta>1\colon\quad &\Delta\big\vert_{h=0}=-(2^\alpha-2)A^\beta<0,\\
    \alpha=1,\;\beta=1\colon\quad &\Delta\equiv 0,\\
    \alpha=1,\;\beta>1\colon\quad &\Delta=J>0.
  \end{aligned}
\end{equation}

In the case $\alpha>1$ and $\beta>1$, the inequality $\Delta(A,0)<0$ implies, by continuity and compactness of $[x/2,x]$, that $\Delta(A,h)<0$ for all $A\in[x/2,x]$ and all sufficiently small $h>0$. More precisely, since $\Delta(A,0)=-(2^\alpha-2)A^\beta$ is continuous and strictly negative on the compact interval $[x/2,x]$, there exists $c_x>0$ such that $\Delta(A,0)\le -c_x<0$ for all $A\in[x/2,x]$. By continuity of $\Delta$, there exists $\varepsilon_0>0$ such that $\Delta(A,h)<0$ for all $A\in[x/2,x]$ and $0<h<\varepsilon_0$.

\proofpart{1. Assume $\beta=1$}
Set $\mu\defsymbol\tfrac12$; the condition $q<p/2$ ensures $p>2q$, so $0$ is
the unique weighted median of both $\pi^\alpha$ and $\tilde\pi^\alpha$.
The characterization of $\MEC[\alpha,1]$ by weighted medians then gives exact formulae:
\[
  \MEC(\pi)=q^\alpha(2x+h),
  \qquad
  \MEC(\tilde\pi)=(2q)^\alpha\!\left(x+\tfrac{h}{2}\right)=2^{\alpha-1}q^\alpha(2x+h).
\]
\proofpart{1.1 Assume $\alpha>1$} Then $2^{\alpha-1}>1$ and so $\MEC(\pi)<\MEC(\tilde\pi)$ for every
$h,q>0$: Axiom~1 holds.

\proofpart{1.2 Assume $\alpha=1$} Then $\MEC(\pi)=\MEC(\tilde\pi)$ for every $h,q>0$: Axiom~1 fails.

\proofpart{2. Assume $\beta>1$} First we locate the optimal consensus points.
For any $q<\mu_0 p$, where
\[
  \mu_0\defsymbol\min\!\left\{(1+2^{\beta-1})^{-1/\alpha},\;2^{-(\alpha+\beta-1)/\alpha}\right\}>0,
\]
the conditions $p^\alpha>(1+2^{\beta-1})q^\alpha$ and $p^\alpha>2^{\alpha+\beta-1}q^\alpha$
both hold. Since $h<x/2$ (to be enforced below), we have
$x/2+h<x$ and $x/2+h/2<x$, so
\[
  E'(x/2)
  \;\geq\;
  \beta\!\left(\tfrac x2\right)^{\!\beta-1}\!\!\bigl[p^\alpha-(1+2^{\beta-1})q^\alpha\bigr]>0,
\]
\[
  \widetilde{E}'(x/2)
  \;\geq\;
  \beta\!\left(\tfrac x2\right)^{\!\beta-1}\!\!\bigl[p^\alpha-2^{\alpha+\beta-1}q^\alpha\bigr]>0.
\]
Since $\beta>1$, $\beta p^\alpha z^{\beta-1}\to 0$ as $z\to 0^+$, giving
\[
  E'(0^+)=-\beta q^\alpha x^{\beta-1}-\beta q^\alpha(x+h)^{\beta-1}<0,
\]
\[
  \widetilde{E}'(0^+)=-\beta(2q)^\alpha\!\left(x+\tfrac{h}{2}\right)^{\beta-1}<0.
\]
Combined with the bounds above, we have
\[
  E'(0^+)<0<E'(x/2),
  \qquad
  \widetilde{E}'(0^+)<0<\widetilde{E}'(x/2),
\]
and therefore, by strict convexity, the unique optimal consensus points $z_\pi$ and $z_{\tilde\pi}$ lie in $(0,x/2)$. In particular,
\[
  A_\pi\defsymbol x-z_\pi\in\!\left(\tfrac x2,x\right),
  \qquad
  A_{\tilde\pi}\defsymbol x-z_{\tilde\pi}\in\!\left(\tfrac x2,x\right).
\]

\proofpart{2.1 Assume $\alpha>1$}
Set $\mu\defsymbol\mu_0$. Since $\Delta(A,0)=-(2^\alpha-2)A^\beta<0$ and $\Delta$ is continuous,
compactness of $[x/2,x]$ yields $\varepsilon_0>0$ such that
\[
  \Delta(A,h)<0
  \qquad\text{for all }A\in\bigl[\tfrac x2,x\bigr],\ h\in(0,\varepsilon_0).
\]
Set $\varepsilon\defsymbol\min(\varepsilon_0,x/2)$. This ensures both $h<\varepsilon_0$ (so that the sign
condition on $\Delta$ applies) and $h<x/2$ (so that the location argument
of the previous paragraph applies, giving $A_{\tilde\pi}\in(x/2,x)\subset[x/2,x]$).

Applying identity~\eqref{eq:central} at $z=z_{\tilde\pi}$:
\[
  E(z_{\tilde\pi})-\widetilde{E}(z_{\tilde\pi})=q^\alpha\,\Delta(A_{\tilde\pi},h)<0,
\]
so $E(z_{\tilde\pi})<\widetilde{E}(z_{\tilde\pi})=\MEC(\tilde\pi)$. Since
$\MEC(\pi)\leq E(z_{\tilde\pi})$, Axiom~1 holds.

\proofpart{2.2 Assume $\alpha=1$}
Let $\varepsilon,\mu>0$ be arbitrary. Choose
\[
  0<h<\min\!\left\{\varepsilon,\tfrac x2\right\},
  \qquad
  0<q<\min\!\left\{\mu p,\;\frac{p}{1+2^{\beta-1}}\right\},
\]
and set
\[
  \pi=((p,q,q),(0,x,x+h)),
  \qquad
  \tilde\pi=((p,2q),(0,x+h/2)).
\]
Since $\alpha=1$, the derivatives of $E$ satisfy
\[
  E'(x/2)\;\geq\;
  \beta\!\left(\tfrac x2\right)^{\!\beta-1}\!\bigl[p-(1+2^{\beta-1})q\bigr]>0,
\]
\[
  E'(0^+)=-\beta q\,x^{\beta-1}-\beta q\,(x+h)^{\beta-1}<0.
\]
By strict convexity the unique optimal consensus point $z_\pi$ of $\pi$ lies in $(0,x/2)$, so
$A_\pi\defsymbol x-z_\pi>0$.
By~\eqref{eq:decomp} with $\alpha=1$, we have $\Delta(A,h)=J(A,h)>0$ for all $A,h>0$.
Applying identity~\eqref{eq:central} at $z=z_\pi$:
\[
  E(z_\pi)-\widetilde{E}(z_\pi)=q\,J(A_\pi,h)>0,
\]
so $\widetilde{E}(z_\pi)<E(z_\pi)=\MEC(\pi)$. Since $\MEC(\tilde\pi)\leq\widetilde{E}(z_\pi)$,
we conclude $\MEC(\pi)>\MEC(\tilde\pi)$, and Axiom~1 fails.

\medskip
This completes the proof.
\end{proof}

\subsection{Proof of Theorem~\ref{th:mec-axiom-2} (\texorpdfstring{\ERLabel}{ER} Axiom~2)}
\label{appendix:proof-mec-axiom-2}

\mecAxiomTwo*
\begin{proof}
Fix $\alpha\ge 1$ and $\beta\ge 1$. Let $p,q,r>0$ with $p>r$, let $x,y>0$ satisfy $|y-x|<x<y$, and let $0<\Delta<y-x$. Define
\[
\pi=\big((p,q,r),(0,x,y)\big),\qquad
\pi_\Delta=\big((p,q,r),(0,x+\Delta,y)\big).
\]
Since $0<\Delta<y-x$ we have $x+\Delta\notin\{0,x,y\}$, hence $\pi_\Delta=\pi_{x\to x+\Delta}$ in the sense of Def.~\ref{def:movement}. We treat the cases $\beta>1$ and $\beta=1$ separately.

\proofpart{Case $\beta>1$}
Set $f(t):=|t|^{\beta-1}\sign(t)$. The effort function of $\pi$,
\[
\EC[\alpha,\beta][\pi](u)
= p^\alpha|u|^\beta+q^\alpha|u-x|^\beta+r^\alpha|u-y|^\beta,
\]
has derivative
\[
F_\pi(u):=\frac{\partial}{\partial u}\EC[\alpha,\beta][\pi](u)
=\beta\big(p^\alpha f(u)+q^\alpha f(u-x)+r^\alpha f(u-y)\big).
\]
Evaluating at $u=x$,
\begin{align*}
F_\pi(x)
&=\beta\big(p^\alpha f(x)+q^\alpha f(0)+r^\alpha f(x-y)\big) \\
&=\beta\big(p^\alpha x^{\beta-1}-r^\alpha (y-x)^{\beta-1}\big)>0,
\end{align*}
because $p^\alpha>r^\alpha$ (since $\alpha\ge1$ and $p>r$) and
$x^{\beta-1}>(y-x)^{\beta-1}$ (since $\beta>1$ and $x>y-x$).

Since $\beta>1$, $\EC[\alpha,\beta][\pi]$ is strictly convex, so $F_\pi$ is strictly increasing and has a unique zero $y^*(\pi)$. The inequality $F_\pi(x)>0$ implies $y^*(\pi)<x<x+\Delta$, so $\pi_\Delta=\pi_{x\to x+\Delta}$ is a shift away from consensus. Strict convexity also gives uniqueness of $y^*(\pi)$, so the shift is not a median-shift. By Theorem~\ref{thm:mec-increases:body},
\[
\MEC(\pi_\Delta)>\MEC(\pi).
\]

\proofpart{Case $\beta=1$}
By Cor.~\ref{cor:mec-special-cases}, the $(\alpha,1)$-optimal consensus points of $\pi$ are the medians of the renormalized distribution $\overline{\pi^\alpha}$, which has masses proportional to $(p^\alpha,q^\alpha,r^\alpha)$ at the support points $(0,x,y)$. Since $\alpha\ge 1$ and $p>r$ give $p^\alpha>r^\alpha$, the inequality $r^\alpha>p^\alpha+q^\alpha$ cannot hold, so $y$ is never a median. Hence every optimal consensus point of $\pi$ lies in $[0,x]$. The mass profile $(p^\alpha,q^\alpha,r^\alpha)$ is preserved by the shift, so the same conclusion holds for $\pi_\Delta$.

We split into two sub-cases according to whether $\pi_{x\to x+\Delta}$ is a median-shift in the sense of Def.~\ref{def:movement}.

\emph{Sub-case 1: $\pi_{x\to x+\Delta}$ is not a median-shift.}
Since every optimal consensus point $y^*$ of $\pi$ satisfies $y^*\le x<x+\Delta$, the shift $\pi_{x\to x+\Delta}$ is away from consensus by Def.~\ref{def:movement}. Theorem~\ref{thm:mec-increases:body} therefore applies and yields $\MEC(\pi_\Delta)>\MEC(\pi)$.

\emph{Sub-case 2: $\pi_{x\to x+\Delta}$ is a median-shift.}
By Def.~\ref{def:movement}, this means that $x$ is the unique $(\alpha,1)$-optimal consensus point of $\pi$, equivalently the unique median of $\overline{\pi^\alpha}$. Concretely, this corresponds to the strict inequalities
\[
p^\alpha<q^\alpha+r^\alpha
\quad\text{and}\quad
p^\alpha+q^\alpha>r^\alpha.
\]
The same mass conditions hold for $\overline{\pi_\Delta^\alpha}$, so $x+\Delta$ is the unique median of $\overline{\pi_\Delta^\alpha}$, hence the unique optimal consensus point of $\pi_\Delta$. By Cor.~\ref{cor:mec-special-cases},
\begin{align*}
\MEC(\pi)
&=\EC[\alpha,1][\pi](x)
=p^\alpha\,x+r^\alpha(y-x),\\
\MEC(\pi_\Delta)
&=\EC[\alpha,1][\pi_\Delta](x+\Delta)
=p^\alpha(x+\Delta)+r^\alpha\big(y-x-\Delta\big).
\end{align*}
Subtracting,
\[
\MEC(\pi_\Delta)-\MEC(\pi)
=\big(p^\alpha-r^\alpha\big)\,\Delta>0,
\]
since $p>r$ and $\alpha\ge 1$ give $p^\alpha>r^\alpha$, and $\Delta>0$.

In either sub-case, $\MEC(\pi_\Delta)>\MEC(\pi)$, which completes the proof.
\end{proof}

\subsection{Proof of Theorem~\ref{th:mec-axiom-3} (\texorpdfstring{\ERLabel}{ER} Axiom~3)}
\label{appendix:proof-mec-axiom-3}

\mecAxiomThree*
\begin{proof}
Fix $\alpha\ge 1$ and $\beta\ge 1$. Fix $p,q>0$ and $x>0$, let $y=2x$, and choose any $\Delta\in(0,q/2)$. Define
\[
\pi := \big((p,q,p), (0,x,2x)\big), \qquad
\tilde{\pi} := \big((p+\Delta, q-2\Delta, p+\Delta), (0,x,2x)\big).
\]
For $\rho\in\{\pi,\tilde{\pi}\}$, consider the effort function
\[
\EC[\alpha,\beta][\rho](u):=\sum_{t\in\{0,x,2x\}} \rho(t)^\alpha\,|u-t|^\beta.
\]
Since $\beta\ge1$, the map $u\mapsto \EC[\alpha,\beta][\rho](u)$ is convex. Moreover, because $\rho(0)=\rho(2x)$ for both
$\rho=\pi$ and $\rho=\tilde{\pi}$, both distributions are symmetric about $x$. By Lemma~\ref{lem:symmetry-mec},
\[
\MEC(\rho)=\EC[\alpha,\beta][\rho](x)\qquad\text{for }\rho\in\{\pi,\tilde{\pi}\}.
\]
Evaluating at $u=x$ yields
\[
\MEC(\pi)=\EC[\alpha,\beta][\pi](x)
=p^\alpha|0-x|^\beta+q^\alpha|x-x|^\beta+p^\alpha|2x-x|^\beta
=2p^\alpha x^\beta,
\]
and similarly
\[
\MEC(\tilde{\pi})=\EC[\alpha,\beta][\tilde{\pi}](x)
=2(p+\Delta)^\alpha x^\beta.
\]
Since $\Delta>0$ and $\alpha\ge1$, we have $(p+\Delta)^\alpha>p^\alpha$. Multiplying by $2x^\beta>0$ gives
\[
2p^\alpha x^\beta < 2(p+\Delta)^\alpha x^\beta,
\]
and thus $\MEC(\pi)<\MEC(\tilde{\pi})$.
\end{proof}

\subsection{Proof of Theorem~\ref{th:mec-property-1} (Property~1)}
\label{appendix:proof-mec-property-1}

\mecPropertyOne*
\begin{proof}
Fix $\alpha\ge 1$ and $\beta\ge 1$. Fix the parameters $p,q,x,z,\Delta$ as in Axiom~\ref{axiom:property-1}, with $y=2x$, and define
\[
\pi := \big((p,q,p),(0,x,2x)\big),
\]
\[
\tilde{\pi} := \big((p,\Delta,q-2\Delta,\Delta,p),(0,x-z,x,x+z,2x)\big).
\]
For $\rho\in\{\pi,\tilde{\pi}\}$, consider the effort function
\[
\EC[\alpha,\beta][\rho](u):=\sum_{t\in \support(\rho)} \rho(t)^\alpha\,|u-t|^\beta.
\]
Since $\beta\ge1$, the map $u\mapsto \EC[\alpha,\beta][\rho](u)$ is convex. Both $\pi$ and $\tilde{\pi}$ are symmetric about $x$ (their masses are paired at equal distance to the left and right of $x$). Therefore, by Lemma~\ref{lem:symmetry-mec},
\[
\MEC(\rho)=\EC[\alpha,\beta][\rho](x)\qquad\text{for }\rho\in\{\pi,\tilde{\pi}\}.
\]

Evaluating at $u=x$ yields
\[
\MEC(\pi)=\EC[\alpha,\beta][\pi](x)
=p^\alpha|x-0|^\beta+q^\alpha|x-x|^\beta+p^\alpha|x-2x|^\beta
=2p^\alpha x^\beta,
\]
and
\begin{align*}
\MEC(\tilde{\pi})
&=\EC[\alpha,\beta][\tilde{\pi}](x) \\
&=p^\alpha|x-0|^\beta+\Delta^\alpha|x-(x-z)|^\beta+(q-2\Delta)^\alpha|x-x|^\beta \\
&\quad +\Delta^\alpha|x-(x+z)|^\beta+p^\alpha|x-2x|^\beta \\
&=2p^\alpha x^\beta+2\Delta^\alpha z^\beta.
\end{align*}
Since $\Delta>0$ and $z>0$, we have $2\Delta^\alpha z^\beta>0$, and hence
\[
\MEC(\pi)=2p^\alpha x^\beta \;<\; 2p^\alpha x^\beta+2\Delta^\alpha z^\beta=\MEC(\tilde{\pi}).
\]
\end{proof}

\section{Auxiliary Theorems}
\label{appendix:auxiliary-theorems}

This appendix collects the closed-form propositions used by the Tipping Point method of Section~\ref{sec:tipping-point}.

We begin with an elementary identity used repeatedly below.

\begin{lemma}[Weighted-mean identity]\label{lem:weighted-variance}
Let \(w_1,\dots,w_s\geq 0\) with \(W\defsymbol\sum_{i=1}^s w_i>0\), and let \(x_1,\dots,x_s\in\mathbb{R}\). Set \(\mu\defsymbol\frac{1}{W}\sum_{i=1}^s w_ix_i\). Then
\[
\min_{y\in\mathbb{R}}\sum_{i=1}^s w_i(x_i-y)^2
=\sum_{i=1}^s w_i(x_i-\mu)^2
=\frac{1}{W}\sum_{1\leq i<j\leq s}w_iw_j(x_i-x_j)^2.
\]
\end{lemma}
\begin{proof}
The function \(y\mapsto\sum_i w_i(x_i-y)^2\) is convex with derivative \(-2\sum_i w_i(x_i-y)\), which vanishes at \(y=\mu\); this gives the first equality. For the second,
\[
\sum_i w_i(x_i-\mu)^2
=\sum_i w_ix_i^2-W\mu^2
=\frac{W\sum_i w_ix_i^2-\bigl(\sum_i w_ix_i\bigr)^2}{W},
\]
and symmetrizing in \((i,j)\) the numerator equals
\(\sum_{i,j}w_iw_j(x_i^2-x_ix_j)=\tfrac12\sum_{i,j}w_iw_j(x_i-x_j)^2=\sum_{i<j}w_iw_j(x_i-x_j)^2\).
\end{proof}

\begin{proposition}[General mass transfer to a more extreme point]
\label{prop:mec-general-mass-shift-formulas}
Let \(0\le l<r<r'\le 1\), let \(m,n>0\), and set \(d\defsymbol r-l\), \(e\defsymbol r'-r\), and \(C\defsymbol m^2+n^2\). Let
\[
\pi=((m,n),(l,r)),
\]
and, for \(k\in(0,n]\), let
\[
\pi_{r:k\rightarrow r'}=((m,n-k,k),(l,r,r')).
\]
Then
\[
\MEC[2,2](\pi)=\frac{m^2n^2d^2}{C},
\]
and
\[
\MEC[2,2](\pi_{r:k\rightarrow r'})
=
\frac{
m^2(n-k)^2d^2+m^2k^2(d+e)^2+(n-k)^2k^2e^2
}{
2k^2-2nk+C
}.
\]
Hence, if
\[
\Delta(k)\defsymbol \MEC[2,2](\pi_{r:k\rightarrow r'})-\MEC[2,2](\pi),
\]
then
\[
\Delta(k)=\frac{k\,R_{m,n,d,e}(k)}{Q_{m,n}(k)},
\]
where
\[
R_{m,n,d,e}(k)
=
e^2Ck^3-2e^2nCk^2+\bigl(2d^2m^4+2de\,m^2C+e^2C^2\bigr)k-2d^2m^4n,
\]
and
\[
Q_{m,n}(k)=C(2k^2-2nk+C).
\]
In particular, \(Q_{m,n}(k)>0\) for all \(k\in[0,n]\), so the sign of \(\Delta(k)\) for \(k>0\) is determined by \(R_{m,n,d,e}(k)\).
\end{proposition}

\begin{proof}
By Lem.~\ref{lem:weighted-variance} applied to \(\pi\) with weights \(m^2\) and \(n^2\) at positions \(l\) and \(r\), and using \(r-l=d\),
\[
\MEC[2,2](\pi)=\frac{m^2n^2(r-l)^2}{m^2+n^2}=\frac{m^2n^2d^2}{C}.
\]
Applying the same lemma to \(\pi_{r:k\rightarrow r'}\), whose weights are \(m^2\), \((n-k)^2\), \(k^2\) at positions \(l\), \(r\), \(r'\), with the squared distances \((r-l)^2=d^2\), \((r'-l)^2=(d+e)^2\), \((r'-r)^2=e^2\),
\[
\MEC[2,2](\pi_{r:k\rightarrow r'})
=
\frac{m^2(n-k)^2d^2+m^2k^2(d+e)^2+(n-k)^2k^2e^2}{m^2+(n-k)^2+k^2}.
\]
Since \(m^2+(n-k)^2+k^2=2k^2-2nk+C\), this gives the stated formula for \(\MEC[2,2](\pi_{r:k\rightarrow r'})\).

Set \(Q_{m,n}(k)\defsymbol C(2k^2-2nk+C)\). Putting \(\Delta(k)=\MEC[2,2](\pi_{r:k\rightarrow r'})-\MEC[2,2](\pi)\) over this common denominator,
\[
\Delta(k)
=
\frac{C\bigl[m^2(n-k)^2d^2+m^2k^2(d+e)^2+(n-k)^2k^2e^2\bigr]-m^2n^2d^2(2k^2-2nk+C)}{Q_{m,n}(k)}.
\]
At \(k=0\) the second weight reduces to \(n^2\) and the third weight vanishes, so \(\pi_{r:0\rightarrow r'}=\pi\) as opinion distributions and the numerator vanishes. The numerator is therefore divisible by \(k\), and direct expansion and simplification yield
\[
\Delta(k)=\frac{k\,R_{m,n,d,e}(k)}{Q_{m,n}(k)},
\]
with \(R_{m,n,d,e}(k)\) the cubic polynomial given in the statement.

Finally, completing the square in the inner factor of \(Q_{m,n}\),
\[
2k^2-2nk+C
=2\bigl(k-\tfrac{n}{2}\bigr)^{2}+\bigl(m^2+\tfrac{n^2}{2}\bigr)>0
\quad\text{for all }k\in\mathbb{R}.
\]
Since also \(C=m^2+n^2>0\), we have \(Q_{m,n}(k)>0\) on the whole real line, in particular on \([0,n]\). Therefore, for \(k>0\), the sign of \(\Delta(k)\) is determined by \(R_{m,n,d,e}(k)\).
\end{proof}

\begin{proposition}[Uniqueness of the tipping point]
\label{prop:unique-root-tipping-general}
Let \(m,n,d,e>0\). Then the polynomial \(R_{m,n,d,e}\) has exactly one root in \((0,n)\). Moreover, this root is simple.
\end{proposition}

\begin{proof}
Set \(x\defsymbol n^2/m^2>0\), \(s\defsymbol e/d>0\), and \(t\defsymbol k/n\). Then \(k\in(0,n)\) if and only if \(t\in(0,1)\), and a direct computation gives
\[
R_{m,n,d,e}(k)=d^2m^4n\,\widetilde R_{x,s}(t),
\]
where
\[
\widetilde R_{x,s}(t)
=
s^2x(x+1)t^3-2s^2x(x+1)t^2+\bigl(s^2(x+1)^2+2s(x+1)+2\bigr)t-2.
\]
Since \(d^2m^4n>0\), the roots of \(R_{m,n,d,e}\) in \((0,n)\) are exactly the roots of \(\widetilde R_{x,s}\) in \((0,1)\).

Now \(\widetilde R_{x,s}(0)=-2<0\), while
\[
\widetilde R_{x,s}(1)=s(x+1)(s+2)>0.
\]
Hence \(\widetilde R_{x,s}\) has at least one root in \((0,1)\).

To prove uniqueness, compute the discriminant of \(\widetilde R_{x,s}\) with respect to \(t\):
\[
\Disc_t(\widetilde R_{x,s})=-4s^2x(x+1)\,P(x,s),
\]
where
\[
P(x,s)=(s^6+2s^5)x^5+x^2(Ax^2+Bx+C)+Dx+E,
\]
with
\[
\begin{aligned}
A&=5s^6+14s^5+8s^4,\\
B&=10s^6+36s^5+26s^4-12s^3,\\
C&=10s^6+44s^5+46s^4+8s^3+23s^2,
\end{aligned}
\]
and
\[
\begin{aligned}
D&=5s^6+26s^5+46s^4+52s^3+59s^2+24s,\\
E&=s^6+6s^5+18s^4+32s^3+36s^2+24s+8.
\end{aligned}
\]
Clearly \(A,C,D,E>0\). Moreover,
\[
4AC-B^2
=
4s^6(s+2)\bigl(25s^5+130s^4+212s^3+204s^2+234s+74\bigr)>0.
\]
Thus the quadratic polynomial \(Ax^2+Bx+C\) is positive for every \(x>0\). Since also \((s^6+2s^5)x^5>0\), \(Dx>0\), and \(E>0\), it follows that \(P(x,s)>0\) for all \(x,s>0\). Therefore
\[
\Disc_t(\widetilde R_{x,s})<0.
\]
A cubic with negative discriminant has exactly one real root. Hence \(\widetilde R_{x,s}\) has exactly one real root. Since \(\widetilde R_{x,s}(0)<0<\widetilde R_{x,s}(1)\), that root lies in \((0,1)\). Therefore \(R_{m,n,d,e}\) has exactly one root \(k^\star\in(0,n)\). Since the discriminant is non-zero, this root is simple.
\end{proof}

\subsection*{Proof of Proposition~\ref{prop:mec22-mass-transfer}}

\massTransfer*
\begin{proof}
Let \(w_i=p_i^2\), \(D=\sum_{i=1}^r w_i\), and \(N=\sum_{1\leq i<j\leq r}w_iw_j(x_i-x_j)^2\). By Lem.~\ref{lem:weighted-variance},
\[
\MEC[2,2](\pi)=\frac{N}{D}.
\]
For \(\pi^{(k)}\), set \(x_{r+1}=1\) and
\[
w_i(k)=
\begin{cases}
p_i^2 & i\le r-1,\\
(p_r-k)^2 & i=r,\\
k^2 & i=r+1,
\end{cases}
\]
and let \(D^{(k)}=\sum_{i=1}^{r+1}w_i(k)=\sum_{i=1}^{r-1}p_i^2+(p_r-k)^2+k^2\) and \(N^{(k)}=\sum_{1\leq i<j\leq r+1}w_i(k)w_j(k)(x_i-x_j)^2\). By Lem.~\ref{lem:weighted-variance} again, \(\MEC[2,2](\pi^{(k)})=N^{(k)}/D^{(k)}\).

Set \(Q_\pi(k)\defsymbol D\,D^{(k)}\). Since \(a^2+b^2\geq\tfrac12(a+b)^2\),
\[
D^{(k)}\geq(p_r-k)^2+k^2\geq\tfrac12 p_r^2>0,
\]
so \(Q_\pi(k)>0\) for every \(k\in\mathbb{R}\), in particular on \([0,p_r]\).

Putting \(\MEC[2,2](\pi^{(k)})-\MEC[2,2](\pi)\) over the common denominator \(Q_\pi(k)\),
\[
\MEC[2,2](\pi^{(k)})-\MEC[2,2](\pi)=\frac{N^{(k)}D-N\,D^{(k)}}{Q_\pi(k)}.
\]
The weights \(w_i(k)\) are polynomials in \(k\) of degree at most \(2\), so \(N^{(k)}\) is a polynomial of degree at most \(4\) and \(D^{(k)}\) a polynomial of degree \(2\). Hence the numerator above is a polynomial of degree at most \(4\). At \(k=0\) we have \(w_{r+1}(0)=0\) and \(w_r(0)=p_r^2\), so \(D^{(0)}=D\) and \(N^{(0)}=N\), giving \(N^{(0)}D-N\,D^{(0)}=0\). The numerator is therefore divisible by \(k\), and we may write
\[
\MEC[2,2](\pi^{(k)})-\MEC[2,2](\pi)=\frac{k\,P_\pi(k)}{Q_\pi(k)}
\]
with \(P_\pi(k)\) a polynomial of degree at most \(3\).

Finally, \(P_\pi\) is genuinely cubic. The only contribution to \(k^4\) in \(N^{(k)}D-N\,D^{(k)}\) comes from the cross term \(w_r(k)w_{r+1}(k)(x_r-1)^2D=(p_r-k)^2k^2(1-x_r)^2D\), whose \(k^4\) coefficient is \(D(1-x_r)^2\). Since \(x_r<1\) and \(D>0\),
\[
[k^3]P_\pi(k)=D\,(1-x_r)^2>0,
\]
so \(P_\pi\) has degree exactly \(3\) with positive leading coefficient.
\end{proof}

\end{document}